\documentclass[letterpaper, 11pt]{article}
\usepackage{booktabs} %
\usepackage[ruled]{algorithm2e} %

\SetAlFnt{\small}
\SetAlCapFnt{\small}
\SetAlCapNameFnt{\small}
\SetAlCapHSkip{0pt}
\IncMargin{-\parindent}

\usepackage[leqno]{amsmath}
\usepackage[pagebackref=true]{hyperref}
\hypersetup{
    colorlinks=true,
    linkcolor=red,
    citecolor=blue,
    urlcolor=teal
}
\usepackage{cleveref}
\usepackage[dvipsnames]{xcolor}
\numberwithin{equation}{section}
\usepackage{amsthm}
\usepackage{amssymb}
\usepackage{mathtools}
\usepackage{thmtools}
\usepackage{bbm}
\usepackage{nameref}
\usepackage{commands}
\usepackage{natbib}
\allowdisplaybreaks
\usepackage[margin=1in]{geometry}
\usepackage{xurl}
\usepackage{authblk}
\usepackage{makecell}

\theoremstyle{plain}
\newtheorem*{theorem*}{Theorem}
\newtheorem*{lemma*}{Lemma}
\newtheorem*{proposition*}{Proposition}
\newtheorem*{claim*}{Claim}
\newtheorem{theorem}{Theorem}[section]
\newtheorem{lemma}[theorem]{Lemma}
\newtheorem{proposition}[theorem]{Proposition}
\newtheorem{corollary}[theorem]{Corollary}
\newtheorem{claim}[theorem]{Claim}
\crefname{claim}{claim}{claims}
\Crefname{claim}{Claim}{Claims}

\theoremstyle{definition}
\newtheorem{definition}[theorem]{Definition}

\theoremstyle{remark}
\newtheorem{remark}[theorem]{Remark}

\setcitestyle{authoryear}
\title{The Computable but Not Learnable Information-Value-Free Equilibria and Regulation of Algorithmic Collusion}
\author{Jason D. Hartline\thanks{hartline@northwestern.edu} \quad Chang Wang\thanks{wc@u.northwestern.edu} \quad Chenhao Zhang\thanks{chenhao.zhang.rea@u.northwestern.edu}}
\affil{Northwestern University}
\date{}

\newcommand\blfootnote[1]{%
    \begingroup
    \renewcommand\thefootnote{}\footnote{#1}%
    \addtocounter{footnote}{-1}%
    \endgroup
}

\begin{document}

\maketitle
\blfootnote{The authors are supported in part by NSF-funded Institute for Data, Econometrics, Algorithms and Learning (IDEAL) through the grant NSF ECCS-2216970.}
\begin{abstract}
A correlated equilibrium is \emph{information-value-free} if every player has an action that yields the same payoff as their equilibrium strategy when all other players follow their respective equilibrium strategies. An equilibrium is learnable if the empirical history of play generated by certain learning algorithms using only full feedback on each player's own payoff converges to it.

Our main result is that although an information-value-free equilibrium can be computed efficiently offline, it is not learnable by a broad class of learning algorithms. This separation stands in sharp contrast to canonical equilibrium concepts, where offline computation and online learning typically have comparable difficulty.

In information-economics terms, our results imply that learning correlated equilibria in games cannot avoid generating valuable information, which has implications for current debates on the regulation of algorithmic collusion: Valuable and implicit information exchange is unavoidable under rationalizable learning, and traditional antitrust regulation against such exchange is therefore incompatible with rationalizable learning. Our results also imply an informational impossibility result for time-average convergence to a Nash equilibrium by a broad class of learning algorithms.

\end{abstract}

\section{Introduction}\label{sec:intro}
We study a relaxation of Nash equilibrium and refinement of correlated equilibrium \citep{aumann1974subjectivity}, which we call \emph{information-value-free equilibrium (IVFE)}. Informally, a correlated equilibrium is \emph{information-value-free} if every player has an action that gives the same payoff as her equilibrium strategy when all other players play their respective equilibrium strategies.

The two requirements of an information-value-free equilibrium, correlated equilibrium and information-value-freeness, have parallel equivalent descriptions that appear throughout the paper. These parallel descriptions are listed in \Cref{table:terms} and are useful for understanding the intuition behind the information-value-free equilibrium and the technical results.

\begin{table}[htbp]
\centering
\begin{tabular}{c|c|c|c}
 & \textbf{Equilibrium Concept} & \makecell{\textbf{Property of} \\ \textbf{Learning}} & \makecell{\textbf{Property of} \\ \textbf{Information}} \\ \hline
\textbf{Requirement 1} & Correlated equilibrium & No swap regret & \makecell{Uses all revealed information \\ (i.e., rationalizable)} \\
\textbf{Requirement 2} & \makecell{Information-value- \\ free equilibrium} & \makecell{Non-negative best-\\in-hindsight regret} & Has no valuable information
\end{tabular}
\caption{Parallel equivalent descriptions of the two requirements of an information-value-free equilibrium}
\label{table:terms}
\end{table}

The first column restates the two equilibrium requirements: the first requirement is correlated equilibrium, and the second is information-value-freeness. The learning column can be understood by considering, informally, a setting where the players repeatedly play a game and we would like the empirical history of play to converge to an information-value-free equilibrium. Then, to converge to a correlated equilibrium, the players must have \emph{no swap regret}. Indeed, when the empirical history of play is a correlated equilibrium, no player can benefit by systematically replacing each recommended action with another action, which corresponds to no swap regret. Moreover, for the correlated equilibrium to be information-value-free, no player can get a higher payoff than she could have obtained by unilaterally deviating to a single fixed action, which corresponds to \emph{non-negative best-in-hindsight regret}\footnote{Best-in-hindsight regret is also known as \emph{external} regret in the literature.}.

In a standard information economics framework, satisfying these requirements implies that the learning outcome is \emph{rationalizable} and does not generate valuable information (the third column of \Cref{table:terms}), hence the name ``information-value-free''. Here we use the term \emph{rationalizable learning} to mean learning that fully exploits the payoff-relevant information generated by the history of play, which is the learning version of the Bayesian rationalization of correlated equilibrium in \citet{aumann1987correlated}. Formal definitions of information can be found in \Cref{sec:ff}.

The main research question we ask in this paper is then:
\begin{quote}
\emph{Is valuable information inherent in rationalizable learning in games, or can rationalizable learning avoid generating valuable information?}
\end{quote}
Here we use the terms from the third column, as they are most relevant to the practical context we are interested in. Briefly speaking, by answering the question we challenge certain regulations of algorithmic pricing, some already enacted and some proposed, that require algorithms to be independent and not to create or exchange valuable information. We conclude that these regulations are algorithmically counterproductive. See the first paragraph of \Cref{subsec:imp} for more background and discussion.

For the resolution of the question, there are intuitions supporting both positive and negative answers to it. On the one hand, \emph{in hindsight}, being information-value-free is a benchmark that is no stronger than established attainable benchmarks in online learning, as one could simply play a best-in-hindsight action. On the other hand, learning is usually understood to be a process of exploring and exploiting information. Even when learning algorithms are uncoupled in the canonical sense, i.e., they respond only to full feedback on their own payoffs, they may, in principle, still create endogenous correlations through the dynamics of play. Players can condition on this correlation, extract valuable information from it, and benefit from doing so. 

Ultimately, we obtain the following negative result, where we use the terms from the second column as they are most convenient for the technical analysis.

\begin{theorem*}[Informal]\label{thm:informal}

Fix any no-swap-regret learning algorithm $\mathcal L$ and any opponent algorithm $\mathcal A$ from a broad class of learning algorithms. There is a sequence of games $\{G_T\}_{T=1}^\infty$, such that when the players play in $G_T$ for $T$ rounds, every empirical history of play remains an asymptotically constant distance from an IVFE. The constant depends only on the property of $\mathcal A$.
\end{theorem*}
Here, the broad class of learning algorithms includes algorithms that are \emph{smooth learners} (\Cref{def:smooth}). This concept captures a smooth response to a very small payoff gap, as learning algorithms with statistical guarantees gradually, not abruptly, abandon a slightly worse arm.

\subsection{Implications}\label{subsec:imp}

Our results have several implications, discussed below.

\paragraph{Generation of valuable information in rationalizable learning and regulation of algorithmic collusion}

Most importantly, for regulation of algorithmic collusion, our results point out that the currently proposed regulation of algorithmic collusion that attempts to suppress all (valuable) information exchange is incompatible with learning and thus unproductive.

Algorithmic collusion refers to the concern that the widespread adoption of pricing algorithms could lead to harmful supra-competitive outcomes among competing sellers of goods and services. These concerns are further exacerbated by the fact that the supra-competitive outcomes can be facilitated without overt human communication, thereby falling outside the scope of conventional regulatory frameworks and current antitrust legal doctrines. However, recently enacted \citep{CA_AB325_2025,NYS7882_2025} and currently proposed \citep{USCongress2025S232,dojrealpage2024} remedies remain within traditional legal frameworks and attempt to suppress implicit communication. In contrast, our results show that valuable implicit communication is inherent in rationalizable learning algorithms, thus providing technical evidence of limitations of current antitrust doctrines. To elaborate, implicit information exchange is neither limited to clear-cut reward-punishment schemes \citep{calvano2020artificial} nor an artifact of poorly designed or pathological algorithms \citep{banchio2022artificial}, but a general and unavoidable feature of learning competitive outcomes in strategic environments. Rationalizable learning itself requires allowing the generation of valuable information. Consequently, seeking to suppress all possible channels of communication between algorithms or requiring algorithmic independence \citep{USCongress2025S232,dojrealpage2024,NYS7882_2025,CA_AB325_2025} may be misaligned with the goal of obtaining competitive prices: \emph{regulation against such information is regulation against rationalizable learning}. 

Furthermore, in the study of algorithmic collusion, Nash equilibrium is commonly used as the standard competitive outcome. However, as we discuss in the last paragraph of this subsection, our results provide informational evidence against using this standard in general algorithmic markets.

We discuss the antitrust policy implications further in \Cref{subsec:ac}.

\paragraph{An equilibrium concept with separation between offline computation and online learning} Beyond algorithmic collusion, information-value-free equilibrium is interesting in its own right as an equilibrium concept. It is not hard to show that it can be computed offline in time polynomial in the size of the game's normal-form representation. The impossibility result, however, shows that such an equilibrium cannot be learned by a broad class of learning algorithms. The result stands in sharp contrast to other canonical equilibrium concepts, where offline computation from the normal-form representation and online learning via uncoupled dynamics have comparable difficulty. For example,
computing a correlated equilibrium of a normal-form game offline is a tractable problem. At the same time, a correlated equilibrium is also easily learnable using uncoupled dynamics \citep{foster1997calibrated, hart2000simple}.  %
On the other hand, computing a Nash equilibrium offline is intractable \citep{daskalakis2009complexity, chen2009settling}, and consequently there are no efficient learning algorithms that guarantee convergence to a Nash equilibrium in general games.

\paragraph{Time-average equilibrium convergence}
As mentioned above, a central goal of learning in games is to understand what kinds of long-run behavior or equilibria can arise when players adapt based solely on their own observations. It is well-known that a broad class of uncoupled learning algorithms, most notably those achieving no best-in-hindsight regret or no swap regret, enjoy strong convergence guarantees \citep{foster1997calibrated,hart2000simple}. When all players follow such algorithms, the empirical distribution of historical play converges to a coarse correlated equilibrium or correlated equilibrium, respectively. These guarantees are achieved by computationally efficient algorithms.

At the same time, under the previously mentioned standard assumptions in computational complexity \citep{daskalakis2009complexity, chen2009settling}, there are no learning algorithms such that, when all players use them, the empirical history of play converges quickly to a Nash equilibrium in general games. As a result, there is a computational-complexity gap between the positive convergence results for (coarse) correlated equilibrium and the apparent inaccessibility of a Nash equilibrium.

Our results establish an informational lower bound for Nash equilibrium convergence for a broad class of learning algorithms. This lower bound follows from our new result and the fact that all Nash equilibria are IVFEs.

\subsection{Related Work}\label{subsec:liter}
In this subsection, we discuss additional related literature on the economics of learning algorithms and algorithmic collusion. See also \citet{hartline2026economics} for a comprehensive review.

\paragraph{Learning with non-negative best-in-hindsight regret}
Learning with non-negative best-in-hindsight regret has appeared in previous work in different ways. When studying mechanism design without a prior, \citet{camara2020mechanisms,collina2024efficient} model a learning agent with no information as having non-negative best-in-hindsight regret and assume the agent can satisfy it.\footnote{To be more precise, they require the agent to have no \emph{contextual} swap regret and non-negative \emph{contextual} best-in-hindsight regret where the context includes the principal's mechanism (i.e., action). In our setting, the players cannot observe the other players' actions.} In contrast, our results show that it is not possible to avoid the generation of valuable information.
\citet{blum2018preserving} 
propose an approach that combines expert advice with group fairness guarantees. The approach relies on learning algorithms that guarantee non-negative best-in-hindsight regret (and no best-in-hindsight regret), and, as we show, it would not be successful if the optimization goal were strengthened to no swap regret and non-negative best-in-hindsight regret. \citet{gofer2016lower} use non-negative best-in-hindsight regret as a property in an intermediate step to prove sequence-dependent regret lower bounds for a family of online learning algorithms.

\paragraph{Convergence results for learning dynamics}
A large body of work has been devoted to positive results on convergence for learning algorithms beyond the celebrated results for (coarse) correlated equilibrium \citep{foster1997calibrated,hart2000simple}.

In certain families of games of interest for algorithmic collusion, learning algorithms provably converge to Nash equilibrium. While the papers discussed below establish Nash-convergence guarantees for specific classes of games, our results identify games in which analogous guarantees cannot hold for the family of algorithms considered in their papers.

\citet{deng2022nash} completely characterize the learning dynamics of mean-based algorithms \citep{braverman2018selling} in repeated first-price auctions with fixed-value bidders. Using iterative elimination of dominated strategies, they show that the dynamics converge to a Nash equilibrium when at least two bidders tie for the highest value. \citet{bichler2024online} generalize their results and techniques to other games, such as Bertrand competitions. Beyond mean-based algorithms, \citet{ahunbay2025semicoarse} prove convergence results similar to those of \citet{deng2022nash} for projected gradient ascent in first-price auctions. \citet{bichler2024online} and \citet{ahunbay2025semicoarse} prove their results by defining and studying an induced equilibrium concept called \emph{semicoarse correlated equilibrium} that arises when all players use certain learning algorithms, and then showing that this refinement is a Nash equilibrium in specific games. More recently, \citet{cai2026proximal} introduce \emph{proximal regret} and \emph{proximal correlated equilibria}, which provide a general framework for refined equilibrium guarantees of gradient-based learning that subsumes semicoarse correlated equilibrium as a special case.

In general settings, it is also possible that certain discrete-time learning dynamics find a Nash equilibrium in exponential time. \citet{foster2003learning,foster2006regret,hart2006stochastic,germano2007global} show that learning dynamics based on \emph{guess-and-verify} can
reach a Nash equilibrium or an $\varepsilon$-Nash equilibrium in finite time, almost surely or with high probability. Here, guess-and-verify means the algorithms guess distributions to play and verify whether the distributions form a Nash equilibrium. Recall that our results apply to smooth learners, and consequently do not apply to the guess-and-verify algorithms because their strategy distributions can change abruptly. 

\paragraph{Impossibility results for last-iterate convergence to a Nash equilibrium}
Since our results also rule out the possibility of time-average convergence to Nash for a broad class of learning algorithms, our paper is related to some previous work that establishes impossibility results for convergence to a Nash equilibrium for dynamics. However, this line of work focuses mainly on \emph{continuous-time} dynamics and \emph{last-iterate} convergence (which is a harder goal than time-average convergence). \citet{hart2003uncoupled} show that no uncoupled dynamics guarantees last-iterate convergence to a Nash equilibrium. \citet{toonsi2023higher} generalize the results to higher-order dynamics that are allowed to maintain richer states. More recently, \citet{milionis2023impossibility} show that there exist games in which no deterministic learning dynamics can eventually converge to a Nash equilibrium. For the discrete-time setting, \citet{babichenko2012completely} considers games that have a pure Nash equilibrium and shows that convergence to a pure Nash equilibrium is impossible. Our results relax their assumption of pure Nash equilibrium but restrict attention to smooth learners.

\paragraph{Manipulation of learning algorithms}
Our results are also technically connected to the literature on manipulation of learning algorithms, as they can be viewed as constructing ``bad'' environments for certain learning algorithms. \citet{braverman2018selling} pioneered the study of manipulating an agent who uses no-best-in-hindsight-regret learning algorithms in auctions and introduced the notion of mean-based algorithms. \citet{deng2019strategizing} further generalize these results by studying how to manipulate a no-best-in-hindsight-regret learner in general bimatrix games to achieve payoffs beyond the \emph{Stackelberg value}, and show that no-swap-regret algorithms are resistant to manipulation. \citet{arunachaleswaran2025algorithmic} show that in Bertrand pricing games, adopting simple static strategies against a no-best-in-hindsight-regret learner can also lead to supra-Nash prices. Moreover, \citet{kumar2024strategically} show that gradient-ascent-based algorithms (which are not no-swap-regret algorithms in general) cannot be manipulated in repeated first-price auctions. \citet{cai2026online,zhao2026no} further generalize and extend the techniques in \citet{kumar2024strategically} and reduce strategic robustness to online linear optimization.

\paragraph{Algorithmic collusion}
Our results have implications for the study of algorithmic collusion. \citet{calvano2020artificial} show that Q-learning algorithms with sufficiently rich state spaces to capture historical interactions can learn reward-punishment schemes that sustain supra-competitive prices. On the other hand, 
\citet{banchio2022artificialb,banchio2022artificial,banchio2023adaptive} show that Q-learning with a minimal state space that is incapable of learning a reward-punishment scheme can still achieve supra-competitive states via statistical errors of Q-learning's misspecified model.
Q-learning has also been studied in other settings by \citet{klein2021autonomous, asker2022artificial}, and \citet{asker2023impact}. Other papers have also explored algorithmic collusion beyond Q-learning, such as the UCB algorithm \citep{hansen2021frontiers} and large language models (LLMs) \citep{fish2024algorithmic,keppo2026fragility}. \citet{luo2026algorithmic} empirically show that Q-learning, UCB, and LLMs can converge to supra-competitive prices with a much shorter time horizon than in the previous experiments if they start with pre-trained initial policies. On the other hand, many algorithms guarantee no swap regret and therefore do not collude for a reasonable definition of rationalizable outcomes \citep{hartline2024regulation,hartline2025regulation}. We consider the same definition of rationalizable outcomes in this paper.

One of the main challenges in the regulation of algorithmic collusion is that it arises through the implicit communication channel between the algorithms, whereas overt communication plays a fundamental role in contemporary legal doctrine of (non-algorithmic) collusion \citep{harrington2018developing,chassang2023regulating}. To address this challenge, \citet{hartline2024regulation,hartline2025regulation} propose to audit non-collusion by designing a statistical test that uses only the data collected during deployment. On the other hand, practitioners  have responded to the challenge by proposing bans on (implicit) communication channels or even on the use of algorithms altogether \citep{USCongress2025S232,dojrealpage2024, NYS7882_2025, CA_AB325_2025}. These approaches have been criticized \citep{harrington2025critique,harrington2025challenges}. Our results affirm these critiques by showing that an implicit information channel must be allowed to exist between rationalizable learning algorithms. The proposed policies neglect such inherent features and can preclude algorithms that competitively optimize. 

\subsection{Organization}
We define the information-value-free equilibrium (IVFE) in \Cref{subsec:staticivfe}. We consider the problem of learning it in \Cref{subsec:learningivfe}. We explain how the generation of valuable information relates to violations of the information-value-free property in \Cref{subsec:learningivfe}. These concepts are used to state and interpret the results in \Cref{sec:impos}. In \Cref{sec:impos} we state our main results and illustrate the general idea of the proof. In \Cref{sec:conclusion} we discuss the implications of our results for algorithmic collusion. All formal proofs are deferred to the appendices.

\section{Equilibrium Concept and Learning}\label{sec:ff}

A \emph{normal-form game} specifies a one-shot interaction between $n$ players. It consists of an action set $A_i$ and a payoff function $u_i$ for each player $i$. Each player simultaneously chooses an action $a_i\in A_i$, and the resulting action profile $a=(a_1,a_2,\dotsc,a_n)$ determines the payoffs $u_i(a_1,a_2,\dotsc,a_n)$. We refer to this one-shot game as the \emph{stage game} when it is used as the primitive interaction in a repeated setting.

\subsection{The Information-Value-Free Equilibrium}\label{subsec:staticivfe}
\begin{definition}\label{def:ivfe}
    Consider a two-player game $G$ in which each player $i$ has a finite set of pure strategies $A_i$ and a payoff function $u_i:A_1 \times A_2 \to \mathbb R$.

    An \emph{information-value-free equilibrium} (IVFE) is a distribution $\tau$ over $A_1 \times A_2$ that satisfies:
    \begin{enumerate}
        \item $\tau$ is a correlated equilibrium; that is,
        \[ \ee{a \sim \tau}{u_i(a) \mid a_i} \geq \ee{a \sim \tau}{u_i(a_i', a_{-i}) \mid a_i} \]
        for every $i$, action $a_i \in A_i$ with $\Pr_\tau[a_i]>0$, and deviation $a_i' \in A_i$. (For actions with zero-probability marginals, the constraints are vacuously satisfied.)
        \item $\tau$ is \emph{information-value-free}, i.e., for every player $i$, the equilibrium payoff is not strictly higher than the payoff from some fixed action: there exists $a_i^\ast \in A_i$ such that
        \[ \ee{a \sim \tau}{u_i(a)} \leq \ee{a \sim \tau}{u_i(a_i^\ast, a_{-i})}. \]
    \end{enumerate}
\end{definition}

It is easy to observe that Nash equilibrium satisfies (1) and (2).
\begin{proposition}\label{prop:ne_ivfe_ce}
Let NE, IVFE, CE be the sets of Nash equilibria, information-value-free equilibria, and correlated equilibria. Then $\text{NE} \subsetneq \text{IVFE} \subsetneq \text{CE}$.
\end{proposition}

The next result shows that in a stage game, an IVFE can always be computed offline in polynomial time. Together with the results in \Cref{sec:impos}, we obtain a separation: an IVFE is efficiently computable offline, but cannot be learned online by a broad class of learning algorithms.
\begin{theorem}\label{prop:static_lp}
    An IVFE always exists and can be computed offline in polynomial time.
\end{theorem}
\begin{proof}[Proof sketch]
    Note that if the best fixed action $a_i^\ast$ for each player $i$ is known, we can construct a linear program similar to that of a correlated equilibrium and solve it. Then, since there are only $|A_1| \cdot |A_2|$ possible pairs of best fixed actions, we can simultaneously solve the linear program that corresponds to each pair efficiently. At least one of the programs has a solution since an IVFE always exists.
\end{proof}

\subsection{Learning IVFE}\label{subsec:learningivfe}

In this subsection, the players repeatedly play the stage game and we formalize the problem of learning an IVFE.

A \emph{repeated game} consists of repeated play of the same stage game over a sequence of rounds. At each round, each player simultaneously chooses an action and receives the corresponding stage-game payoff. Note that the player's choice of action in each round may depend on the history generated so far.

Analogous to the relationship between no-swap-regret learning and correlated equilibrium, and between no-best-in-hindsight-regret learning and coarse correlated equilibrium, learning an IVFE means achieving no swap regret and non-negative best-in-hindsight regret. We then connect this notion to the generation of valuable information during learning, using a standard framework from information economics.

\paragraph{Setup} Consider a player with a known action set $A_1$ of size $n$ and a time horizon of $T$ rounds. A stage game is repeated in each round.

In round $t = 1,2,\dotsc,T$:
\begin{enumerate}
    \item The player selects a probability distribution $p^t$ over the $n$ actions.
    \item The reward vector $r^t \in [-c,c]^n$ is induced by the behavior of the other players\footnote{For analytical clarity, we assume the entries are expectations with respect to the opponent's strategy. This assumption is without loss of generality: if needed, standard concentration arguments can remove this assumption at the cost of additional technical complexity. The same assumption is also standard in classical results, such as the derivation of correlated equilibrium guarantees from no-swap-regret algorithms.}; for simplicity, we refer to the other players collectively as the opponent. Here, $c>0$ is an absolute constant.
    \item An action $a \in A_1$ is then sampled from $p^t$, and the player receives reward $r^t_a$.
    \item The player observes the entire reward vector $r^t$ (which is known as \emph{full feedback}).
\end{enumerate}
The player's learning algorithm maps the history of past probability distributions and reward vectors to the probability distribution for the next round.

As discussed above, learning an IVFE amounts to achieving no swap regret and non-negative best-in-hindsight regret for both players. One way to achieve an IVFE is for each player to unilaterally guarantee the required properties on their side. We begin by recalling the two regret notions.

\begin{definition}\label{def:swap}
    The player's (per-round) \emph{swap regret} is defined as
    \[ \text{SwapRegret}_T := \max_{\sigma:[n] \to [n]} \frac 1T \sum_{t=1}^T \ee{a \sim p^t}{r^t_{\sigma(a)}} - \frac 1T \sum_{t=1}^T \ee{a \sim p^t}{r^t_a}. \]
    The player's algorithm has \emph{no swap regret} if $\text{SwapRegret}_T=o(1)$ as $T \to \infty$.

\end{definition}
\begin{definition}\label{def:best-in-hindsight}
    The player's (per-round) \emph{best-in-hindsight regret} is defined as
    \[ \text{BIH-Regret}_T := \max_{a' \in [n]} \frac 1T \sum_{t=1}^T r^t_{a'} - \frac 1T \sum_{t=1}^T \ee{a \sim p^t}{r^t_a}. \]
    The player's algorithm has \emph{non-negative best-in-hindsight regret} if $\text{BIH-Regret}_T \geq -o(1)$ as $T \to \infty$.
\end{definition}
Combining no swap regret with non-negative best-in-hindsight regret yields the following definition.
\begin{definition}\label{def:zer}
    A no-swap-regret player is \emph{information-value-free} if she achieves non-negative best-in-hindsight regret.
\end{definition}
\begin{remark}
    \emph{In hindsight}, the property is trivial to satisfy: always playing the best-in-hindsight action yields no swap regret and non-negative best-in-hindsight regret. The impossibility is to achieve these guarantees through (rationalizable) learning.
\end{remark}

If both players have no swap regret and non-negative best-in-hindsight regret, then the empirical distribution of play \emph{converges} to an information-value-free equilibrium. To formalize convergence, we first define a measure of how far a distribution is from satisfying the IVFE conditions. This definition can be seen as the regret notion for the static setting.
\begin{definition}
    Consider the stage-game setup in \Cref{def:ivfe}. For a game $G$ and any distribution $\tau$ over $A_1 \times A_2$,
    let
    \[ d_{G, i}^{CE}(\tau) = \ee{a_i \sim \tau_{i}}{\max_{a_i'\in A_i}\ee{a \sim \tau}{u_i(a_i', a_{-i}) \mid a_i} -
    \ee{a \sim \tau}{u_i(a) \mid a_i}}
    \]
    where $\tau_i$ is the $i$-th marginal of $\tau$ , 
    and
    \[
    d_{G,i}^{IVF}(\tau) = \left(\ee{a \sim \tau}{u_i(a)} - \max_{a_i^\ast \in A_i}\ee{a \sim \tau}{u_i(a_i^\ast, a_{-i})}\right)_+,
    \]
    be the distance of player $i$ from satisfying the two conditions in the definition respectively. Define $\tau$'s \emph{distance to an IVFE} as $d_G(\tau) = \max_i\left\{\max\left\{d_{G,i}^{IVF}(\tau), d_{G,i}^{CE}(\tau)\right\}\right\}$.
\end{definition}

\begin{proposition}\label{prop:ivfe-convergence}
    Let $p^t$ and $(p')^t$ denote the player's and opponent's strategies in round $t$, respectively. Let $\tau^t = p^t \times (p')^t$. If both players have no swap regret and are information-value-free, then the empirical history of play $\tau_T=\frac{1}{T} \sum_{t=1}^T \tau^t$ converges to an information-value-free equilibrium:
    \[ \lim_{T \to \infty} d_G(\tau_T)=0. \]
\end{proposition}

\paragraph{Interpretations from information economics} In the remainder of this subsection, we follow the ideas from \citet{aumann1987correlated} to relate \Cref{def:zer} to the generation of valuable information after the learning process (cf. \Cref{table:terms}).

We first define information using the standard framework of \emph{information structure} from information economics.
In this framework, a decision maker faces uncertainty about the payoff-relevant state. Information is modeled as a signal, correlated with the uncertain state, that the decision maker can receive and that conveys payoff-relevant information. The decision maker is rationalizable if she fully exploits her information given by the signal. The information is \emph{valuable} if the decision maker strictly improves her payoff by making decisions conditional on the signal realization she receives.
\begin{definition}[Information Structure]
    An \emph{information structure} consists of two components:
    \begin{itemize}
        \item a set of \emph{states} $\Omega$,
        \item a joint distribution of the \emph{signal} and the state in $\Delta(\Sigma \times \Omega)$.
    \end{itemize}
\end{definition}
We instantiate the definition in our setting as follows when the player is playing a stage game with the opponent.
Consider the player as a decision maker. Given a joint distribution $\tau \in \Delta(A_1\times A_2)$ over the player's actions and the opponent's actions, 
$(A_2,\tau)$ is an information structure with the states $\Omega=A_2$ and signal values $\Sigma = A_1$\footnote{As in the standard modeling of information in games, for any stage game, the detailed payoff structure is constant and involves no uncertainty despite being unknown to the learning algorithms, and is therefore not treated as part of the state.}.

We now explain the connection between the player's information and the properties of her learning algorithm.
\begin{lemma}\label{lemma:ivfconnection}
    If the player has no swap regret and non-negative best-in-hindsight regret, then her learning is rationalizable and does not generate valuable information, where information is defined under the above framework.
\end{lemma}
The formal proof of the lemma is given in \Cref{subsec:connectionproof}. Its key intuition is as follows: First, a player with no swap regret is rationalizable. Second, if the player has non-negative best-in-hindsight regret, she does not benefit from the information.

Below, we elaborate the intuition with  a discussion of information structures. Readers familiar with information structures can skip to the next section.

First, recall that by the definition of no swap regret, the player behaves optimally conditional on the actions she takes. 
The actions of the player can be interpreted as the signal values.
Therefore, when the player's algorithm has no swap regret, given the empirical joint distribution $\tau^T$ over the player's and opponent's actions, she behaves optimally conditional on the signal values. Therefore, the player is rationalizable, i.e., fully exploiting the information revealed by her actions.

Second, if the player has non-negative best-in-hindsight regret, she does not benefit from the information. To see why, recall that by the definition of non-negative best-in-hindsight regret, she is not better off than playing a fixed action.
In the framework of information structure, when choosing a fixed action, the signal can be safely ignored. So when the player's algorithm has non-negative best-in-hindsight regret, given the empirical joint distribution $\tau^T$ over the player's and opponent's actions, her payoff is not higher than what she could have attained without the signal. Therefore, the player does not benefit from the information revealed by her actions.

\section{The Impossibility Result}\label{sec:impos}
In this section, we prove our main impossibility result. We consider two possibilities under which an IVFE may be learned. The stronger possibility is that one may design a no-swap-regret learning algorithm that is information-value-free regardless of what the opponent is doing, for any stage game structure. The weaker possibility is that one may design such a learning algorithm, but information-value-freedom holds only when all players jointly use the algorithm.

We rule out the first possibility in \Cref{coro:impos-ivfe}, and the second for a broad class of \emph{smooth learners} in \Cref{thm:impos-ivfe}. As the name suggests, a smooth learner changes its action probabilities gradually, which is formally defined as follows.
\begin{definition}\label{def:smooth}
    Consider two arms with rewards of zero and $\eps$ and a learning algorithm. Let $w_t$ be the probability of playing the zero-reward arm in round $t$.\footnote{Without loss of generality, assume $w_1 = \Omega(1)$ because (a) either $w_1$ or $1-w_1$ is $\Omega(1)$ and (b) the algorithm cannot distinguish the two arms prior to the first round.} The algorithm is a \emph{smooth learner} if for any time horizon $T$ there exists $\eps = \eps(T)$ such that the following hold:
    
    There exists a function $h(T)=o(1)$ and a weakly decreasing subsequence of $\{w_t\}_{t=1}^T$ of length $T(1-h(T))$, such that the sequence also satisfies
    \[ \sum_{t=1}^{T/2}w_t - \sum_{t=T/2+1}^Tw_t = \Omega(T). \]
\end{definition}

\Cref{def:smooth} says that when faced with two arms with a small difference, the learning algorithm gradually shifts toward the higher-reward arm, but the learning process can be made arbitrarily slow as long as the difference between the two arms $\eps$ is sufficiently small.

The definition captures a smooth response to a very small payoff gap. Learning algorithms with statistical guarantees generally do not instantly abandon the slightly worse arm but exhibit gradual adjustment. Their probabilities move gradually because updates are driven by accumulated payoff differences or small gradients. The definition rules out non-smooth behavior (e.g., an immediate jump to the better action when the gap is chosen small enough). For example, it rules out algorithms in \citet{foster2003learning,foster2006regret,hart2006stochastic,germano2007global}. These \emph{guess-and-verify-based} algorithms can even \emph{inefficiently} learn a Nash equilibrium but are not smooth learners.

We show that many learning algorithms satisfy the smooth-learner property. The first two are mean-based \citep[Appendix D]{braverman2018selling}, whereas the third is not \citep[Appendix B]{ahunbay2025semicoarse}.

\begin{proposition}\label{prop:smooth_example}
    The standard versions of the following algorithms are smooth learners:
    \begin{enumerate}
        \item Multiplicative Weights Update,
        \item Follow-the-Perturbed-Leader,
        \item Online Gradient Ascent (also known as Projected Gradient Ascent \citep{zinkevich2003online}).
    \end{enumerate}
    
    The corresponding Blum--Mansour reductions\footnote{Recall that the Blum--Mansour reduction is a standard method of identifying no-swap-regret algorithms, which provides a general way of reducing no-best-in-hindsight-regret algorithms to no-swap-regret algorithms. In the proof of \Cref{prop:smooth_example} in \Cref{subsec:pf_smooth_example}, we give a brief technical review of the Blum–Mansour reduction.} \citep{blum2007external} of these algorithms are also smooth learners.
\end{proposition}

We are now ready to state the main result.
\begin{theorem}\label{thm:impos-ivfe}
    Fix any learning algorithm $\mathcal L$ with no swap regret against any smooth learner $\mathcal A$. There exists a constant $c>0$ (depending on the smooth-learning property of $\mathcal A$) and a sequence of games $\{G_T\}_{T=1}^\infty$ such that, when the players play $G_T$ for $T$ rounds, the empirical history of play $\tau_T$ remains an asymptotically constant distance from an IVFE:
    \[ \liminf_{T \to \infty} d_{G_T}(\tau_T) \geq c. \]
\end{theorem}
From the information-economics perspective, the result informally says: rationalizable learning generates valuable information.

From the proof of \Cref{thm:impos-ivfe}, we also obtain the following:
\begin{corollary}\label{coro:impos-ivfe}
    There exists a stage game such that for any time horizon $T$, no adversarially robust no-swap-regret learning algorithm can be information-value-free unilaterally (i.e., regardless of the opponent's behavior).
\end{corollary}

By an argument similar to that of \Cref{prop:ne_ivfe_ce}, any $c$-Nash equilibrium is also within $c$ distance to an IVFE. Therefore, we also obtain the following corollary:
\begin{corollary}\label{coro:nashimpos}

    Fix any learning algorithm $\mathcal L$ with no swap regret against any smooth learner $\mathcal A$. There exists a constant $c>0$ (depending on the smooth-learning property of $\mathcal A$) and a sequence of games $\{G_T\}_{T=1}^\infty$ such that when the players play $G_T$ for $T$ rounds, the empirical history of play $\tau_T$ is not a $c$-Nash equilibrium of $G_T$ for every sufficiently large $T$.
\end{corollary}

The full proof of \Cref{thm:impos-ivfe} can be found in \Cref{subsec:proof-impos-ivfe}. Below, we illustrate the main idea in a simplified setting.
\begin{proof}[Proof sketch]
\phantom{filler}

\par

\textbf{Construction.}
    Fix any time horizon $T$ and a player using any no-swap-regret learning algorithm. Suppose the player faces an adversarial Nature instead of another player using a smooth learner.
    
    Consider the setup shown in the table below: Nature can choose between two states and the player has three actions 1, 2, and 3. In one state, the three actions have payoffs $\frac 12, -1, 0$ and in the other state the payoffs are $-1, 1, 0$, respectively.

\begin{table}[htbp]
        \centering
        \begin{tabular}{|c|c|c|c|}
            \hline
                  & Action 1 & Action 2 & Action 3 \\ \hline
            State 1 & $1/2$  & $-1$ & $0$ \\ \hline
            State 2 & $-1$  & $1$ & $0$ \\ \hline
        \end{tabular}
    \end{table}

\textbf{Analysis.}
    For simplicity, assume the player's algorithm is adversarially robust. Furthermore, assume her actions are deterministic, so we can directly count the rounds in which each action is played, and her swap regret is (exactly) non-positive and her best-in-hindsight regret is (exactly) non-negative to avoid any accounting for asymptotic slack.
    
    Nature now sets the first $T/2$ rounds to be in the first state, and the remaining $T/2$ rounds to be in the second state. We show that the player's total payoff is at least $T/4$. Note that the best-in-hindsight payoff is zero, so the best-in-hindsight regret is at most $-T/4$. To do so, we pin down the player's behavior with the following three claims:
    \begin{enumerate}
    \item The player must play action 1 in state 1.
    \item The player can play action 1 for at most $T/4$ rounds in state 2.
    \item The player cannot play action 3 in state 2, so she plays action 2 for at least $T/4$ rounds.
  \end{enumerate}
    
    For the first claim that the player must play action 1 in state 1, note that the player satisfies no swap regret even with an adversary, so during the first $T/2$ rounds, she must have played $T/2$ rounds of action 1. Otherwise, Nature can continue to select the first state and the player would have positive swap regret. 
    
    For the second claim that the player can play action 1 for at most $T/4$ rounds in state 2, let $t$ be the number of rounds she plays action 1 in the remaining $T/2$ rounds. Consider the potential swap from action 1 to action 3. Had she made the swap, the payoff would be zero. And her total payoff when playing action 1 is $\frac 12 \cdot T/2-t$. Therefore, by the no-swap-regret guarantee, $\frac 12 \cdot T/2-t \geq 0$ and $t \leq T/4$.
    
    For the third claim that the player cannot play action 3 in state 2 and therefore plays action 2 for at least $T/4$ rounds, suppose for contradiction that the player plays action 3. Then the swap from action 3 to action 2 contributes to positive swap regret, which is a contradiction. Therefore, the remaining $T/2-t$ rounds must be in action 2, and action 2 is played for at least $T/4$ rounds since $t \leq T/4$ by the previous paragraph.
    
    Combining the three claims above, the player's total payoff is $\frac{1}{2} \cdot T/2-t+(T/2-t) = \frac 34T-2t \geq T/4>0$, where zero is the payoff of the best-in-hindsight action. Therefore, the distance to an IVFE is at least $1/4$.

    In the full proof, we construct a stage game and show that, as long as the other player uses a smooth learner (\Cref{def:smooth}), she behaves similarly to the adversarial Nature described above, and we carefully analyze their behavior to obtain the final impossibility result.
\end{proof}

\section{Discussion}\label{sec:conclusion}

\subsection{Algorithmic Collusion}\label{subsec:ac}

One of the most important implications of our results concerns algorithmic collusion. \citet{calvano2020artificial} first demonstrate that Q-learning algorithms with rich enough state spaces, when run against each other, learn reward-punishment schemes and produce supra-competitive prices. Similar supra-competitive behaviors have been observed for other algorithms and environments. A concerning feature of these behaviors is that the algorithms maintain supra-competitive prices without being explicitly programmed to learn such
strategies or to \emph{communicate explicitly}. This feature creates significant difficulties for regulators because, in the context of antitrust, many jurisdictions outlaw only explicit agreements to coordinate prices in the form of overt communication. However, with algorithms, harmful coordination on prices can happen implicitly through market interactions. For example, in \citet{calvano2020artificial}, the reward-punishment scheme learned by the Q-learning algorithm is never supplied as input to the algorithm or hard-coded into it. Of course, such schemes would be illegal if announced explicitly by the firms.

A tempting way to protect against supra-competitive outcomes in the presence of algorithms is to also rule out any implicit information exchange. Indeed, many regulators in practice suggest remedies by completely banning algorithms \citep{CA_AB325_2025,NYS7882_2025,dojrealpage2024} or banning all communication channels between algorithms \citep{USCongress2025S232}.

However, the currently proposed remedies have been criticized \citep{harrington2025critique,harrington2025challenges}. Our results add to this critique from a technical and structural perspective. In particular, our results imply that implicit information exchange is an inevitable feature of rationalizable learning that optimizes and adapts to the environment. While one could argue that our worst-case construction might not reflect representative scenarios, our results imply that regulators cannot implement a blanket ban on implicit informational exchange, since it would inevitably and adversely affect environments where optimization is necessary for algorithms to find competitive outcomes. Policies that neglect such inherent features can be counterproductive, as they could preclude optimization algorithms that find competitive outcomes. 

Several implications for antitrust policy are in order. First, the legal focus on overt communication is inherently counterproductive in markets mediated by learning algorithms, as it overlooks the fundamental mechanism through which rationalizable learning occurs. On this view, the policy lesson is that not all information exchange among learning algorithms should be banned. Rather, regulation should distinguish information that is necessary for proper optimization from information that predictably enables supra-competitive coordination. This raises the question of how to distinguish permissible learning and adaptation from harmful collusion. We argue that a possible avenue is to move beyond communication-based doctrines toward outcome- or mechanism-based approaches, including statistical tests using deployment data, safe harbors for algorithms that pass such tests, as well as stronger scrutiny of shared vendors, common pricing recommendations, and data flows that aggregate competitor-sensitive information \citep{chassang2023regulating, hartline2024regulation, harrington2025challenges}.

Furthermore, \Cref{coro:nashimpos} on time-average convergence to Nash equilibrium also provides insights into how equilibrium behavior should be interpreted in the study of algorithmic collusion. To elaborate, the usual gold standard for competitive outcomes is Nash equilibrium. However, our results provide evidence against this baseline choice from an information perspective.

\subsection{Open Questions}
It remains open whether there exists a strictly weaker condition than \Cref{def:zer} (information-value-free) such that (a) it is unilaterally achievable, and (b) learning dynamics in which all algorithms simultaneously possess such a property guarantee convergence to the (stronger) IVFE. We conjecture that the answer is negative because our results imply that any smooth learner that has no swap regret against smooth learners cannot satisfy the desiderata. However, we believe that it is unlikely to find non-smooth learners with good learning-theoretic guarantees.

Moreover, our impossibility result relies on a non-computational-complexity assumption on the learning algorithms. It remains open whether IVFE can be learned if one relaxes the assumption and instead restricts attention to learning algorithms whose total number of iterations $T$ is polynomial in the size of the normal-form representation of the underlying game.

\bibliographystyle{ACM-Reference-Format}
\bibliography{refs}

\clearpage
\newpage
\section*{Appendices}
\appendix
\section{\texorpdfstring{Omitted Proofs in \Cref{sec:ff}}{Omitted Proofs in Section~\ref*{sec:ff}}}
\subsection{\texorpdfstring{Proof of \Cref{prop:ne_ivfe_ce}}{{Proof of Proposition~\ref*{prop:ne_ivfe_ce}}}}
\begin{proposition*}
Let NE, IVFE, CE be the sets of Nash equilibria, information-value-free equilibria, and correlated equilibria. Then $\text{NE} \subsetneq \text{IVFE} \subsetneq \text{CE}$.
\end{proposition*}

We show every Nash equilibrium is an IVFE but not vice versa.

First, every Nash equilibrium is a correlated equilibrium. Second, for every player $i$, let $\tau_i$ and $\tau_{-i}$ be the strategies of the player and the opponent, respectively. By definition of Nash equilibrium
\[
\ee{a \sim \tau}{u_i(a)} = \ee{a_i\sim \tau_i}{\ee{a_{-i} \sim \tau_{-i}}{u_i(a_i,a_{-i})}} = \max_{\tau_i'} \ee{a_i\sim \tau_i'}{\ee{a_{-i} \sim \tau_{-i}}{u_i(a_i,a_{-i})}}.
\]
By linearity, the maximizer is attained at an extreme point which implies that there exists $a_i^\ast \in A_i$ such that
\[
\max_{\tau_i'} \ee{a_i\sim \tau_i'}{\ee{a_{-i} \sim \tau_{-i}}{u_i(a_i,a_{-i})}} = 1\cdot \ee{a\sim \tau}{u_i(a_i^\ast,a_{-i})},
\]
and hence
\[
\ee{a \sim \tau}{u_i(a)} \leq \ee{a\sim \tau}{u_i(a_i^\ast,a_{-i})}.
\]

To see that the converse is not true,
consider the following two-player game.
\begin{table}[htbp]
    \centering
    \begin{tabular}{|c|c|c|}
        \hline
              & $b_1$ & $b_2$ \\ \hline
        $a_1$ & 5, 1  & 0, 0  \\ \hline
        $a_2$ & 4, 4  & 1, 5  \\ \hline
    \end{tabular}
\end{table}
The game has three Nash equilibria: $(a_1, b_1), (a_2, b_2)$, and the mixed equilibrium $(\frac 12a_1 + \frac 12a_2, \frac 12b_1 + \frac 12b_2)$. An IVFE that is not a Nash equilibrium is $\frac 13(a_1, b_1) + \frac 13(a_2, b_1) + \frac 13(a_2, b_2)$.

\subsection{\texorpdfstring{Proof of \Cref{prop:static_lp}}{Proof of Theorem~\ref*{prop:static_lp}}}
\begin{theorem*}
    An IVFE always exists and can be computed offline in polynomial time.
\end{theorem*}
A mixed Nash equilibrium is a special case of an IVFE, so it always exists.

\paragraph{Algorithm} Let $r_{ij}:=u_1(a_i, a_j)$, $c_{ij}:=u_2(a_i, a_j)$, $n=|A_1|$, and $m=|A_2|$. For every $a_s \in A_1, a_t \in A_2$, we solve the following linear program $(\text{LP}(s, t))$ in variables $p_{ij}\,(1 \leq i \leq n, 1 \leq j \leq m)$
\begin{align*}
    \sum_{j=1}^m p_{ij}(r_{ij}-r_{i'j}) & \geq 0, \quad \forall i, i' \in [n], i \neq i', && \text{(no swap regret for the row player)} \\
    \sum_{i=1}^n p_{ij}(c_{ij}-c_{ij'}) & \geq 0, \quad \forall j, j' \in [m], j \neq j', && \text{(no swap regret for the column player)} \\
    \sum_{i, j} p_{ij}(r_{sj}-r_{ij}) & \geq 0, &&  \hspace*{-6em} \text{(the row player does not outperform action $s$)} \\
    \sum_{i, j} p_{ij}(c_{it}-c_{ij}) & \geq 0, && \hspace*{-6em} \text{(the column player does not outperform action $t$)} \\
    p_{ij} \geq 0 \quad \forall i, j, &\quad \sum_{i, j}p_{ij}=1.
\end{align*}
Output the joint distribution $p_{ij}$ from any linear program that has a feasible solution. There are $nm$ such linear programs, so it runs in polynomial time.

To see that this algorithm is correct, note that the original program (P) of IVFE is
\begin{align*}
    \sum_{j=1}^m p_{ij}(r_{ij}-r_{i'j}) & \geq 0, \quad \forall i, i' \in [n], i \neq i', && \text{(no swap regret for the row player)} \\
    \sum_{i=1}^n p_{ij}(c_{ij}-c_{ij'}) & \geq 0, \quad \forall j, j' \in [m], j \neq j', && \text{(no swap regret for the column player)} \\
    \max_{1 \leq s \leq n} \sum_{i, j} p_{ij}(r_{sj}-r_{ij}) & \geq 0, && \hspace*{-6em} \text{(non-negative best-in-hindsight regret for the row player)} \\
    \max_{1 \leq t \leq m} \sum_{i, j} p_{ij}(c_{it}-c_{ij}) & \geq 0, && \hspace*{-6em} \text{(non-negative best-in-hindsight regret for the column player)} \\
    p_{ij} \geq 0 \quad \forall i, j, &\quad \sum_{i, j}p_{ij}=1.
\end{align*}

If there is a solution to (P), then there exist $s, t$ that correspond to $\argmax_{1 \leq s \leq n} \sum_{i, j} p_{ij}(r_{sj}-r_{ij})$ and $\argmax_{1 \leq t \leq m} \sum_{i, j} p_{ij}(c_{it}-c_{ij})$ and thus it is a feasible solution to $(\text{LP}(s, t))$. If there is a solution to $(\text{LP}(s, t))$, then it is also a feasible solution to (P) as $\max_{1 \leq s \leq n} \sum_{i, j} p_{ij}(r_{sj}-r_{ij}) \geq \sum_{i, j} p_{ij}(r_{sj}-r_{ij}) \geq 0$.

\subsection{\texorpdfstring{Proof of \Cref{prop:ivfe-convergence}}{Proof of Proposition \ref*{prop:ivfe-convergence}}}
\begin{proposition*}
    Let $p^t$ and $(p')^t$ denote the player's and opponent's strategies in round $t$, respectively. Let $\tau^t = p^t \times (p')^t$. If both players have no swap regret and are information-value-free, then $\frac{1}{T} \sum_{t=1}^T \tau^t$ converges to an information-value-free equilibrium:
    \[ \lim_{T \to \infty} d_G \left(\frac{1}{T} \sum_{t=1}^T \tau^t \right)=0. \]
\end{proposition*}
Let the player's and the opponent's action be indexed by $[n]$ and $[m]$ respectively, by definition of no swap regret, we have
\[
\max_{\sigma:[n] \to [n]} \frac 1T \sum_{t=1}^T \ee{a \sim p^t}{r^t_{\sigma(a)}} - \frac 1T \sum_{t=1}^T \ee{a \sim p^t}{r^t_a} = o(1)
\]
and
\[
\max_{\sigma:[m] \to [m]} \frac 1T \sum_{t=1}^T \ee{a' \sim (p')^t}{(r')^t_{\sigma(a')}} - \frac 1T \sum_{t=1}^T \ee{a' \sim (p')^t}{(r')^t_{a'}} = o(1)
\]
as $T \to \infty$, 
where 
\[
r^t_a = \ee{a'\sim (p')^t}{u_1(a,a')}, (r')^t_{a'} = \ee{a\sim p^t}{u_2(a,a')}.
\]
Therefore,
we have 
\[
\max_{\sigma:[n] \to [n]} \frac 1T \sum_{t=1}^T \ee{a \sim p^t}{\ee{a'\sim (p')^t}{u_1(\sigma(a),a')}} - \frac 1T \sum_{t=1}^T \ee{a \sim p^t}{\ee{a'\sim (p')^t}{u_1(a,a')}} = o(1)
\]
and
\[
\max_{\sigma:[m] \to [m]} \frac 1T \sum_{t=1}^T \ee{a' \sim (p')^t}{\ee{a\sim p^t}{u_2(a,\sigma(a'))}} - \frac 1T \sum_{t=1}^T \ee{a' \sim (p')^t}{\ee{a\sim p^t}{u_2(a,a')}} = o(1).
\]
By definition we have $\tau^t = p^t\times (p')^t$ and let $\tau = \frac 1T \sum_{t=1}^T\tau^t$. Therefore, rearranging and by linearity,
we have as $T \to \infty$,
\[ 
d_{G, 1}^{CE}(\tau) = \ee{a \sim \tau_1}{\max_{a^\ast \in [n]}\ee{(a,a') \sim \tau}{u_1(a^\ast, a') \mid a} -
    \ee{(a,a') \sim \tau}{u_1(a,a') \mid a}} = o(1)
\]
and
\[ 
d_{G, 2}^{CE}(\tau) = \ee{a' \sim \tau_2}{\max_{a^\ast \in [m]}\ee{(a,a') \sim \tau}{u_2(a, a^\ast) \mid a'} -
    \ee{(a,a') \sim \tau}{u_2(a,a') \mid a'}} = o(1).
\]

By definition of information-value-free for both players, we have 
\[
\max_{a^\ast \in [n]} \frac 1T \sum_{t=1}^T r^t_{a^\ast} - \frac 1T \sum_{t=1}^T \ee{a \sim p^t}{r^t_a} \geq  -o(1)
\]
and
\[
\max_{a^\ast \in [m]} \frac 1T \sum_{t=1}^T (r')^t_{a^\ast} - \frac 1T \sum_{t=1}^T \ee{a' \sim (p')^t}{(r')^t_{a'}} \geq -o(1)
\]
as $T \to \infty$.
Therefore, we have
\[
\max_{a^\ast \in [n]} \frac 1T \sum_{t=1}^T \ee{a'\sim (p')^t}{u_1(a^\ast,a')} - \frac 1T \sum_{t=1}^T \ee{a \sim p^t}{\ee{a'\sim (p')^t}{u_1(a,a')}} \geq  -o(1)
\]
and
\[
\max_{a^\ast \in [m]} \frac 1T \sum_{t=1}^T \ee{a\sim p^t}{u_2(a,a^\ast)} - \frac 1T \sum_{t=1}^T \ee{a' \sim (p')^t}{\ee{a\sim p^t}{u_2(a,a')}} \geq -o(1).
\]
Rearranging and by linearity, we have as $T \to \infty$,
\[ 
\ee{(a,a') \sim \tau}{u_1(a,a')} - \max_{a^\ast \in [n]}\left\{\ee{(a,a') \sim \tau}{u_1(a^\ast, a')}\right\} = o(1)
\]
and
\[ 
\ee{(a,a') \sim \tau}{u_2(a,a')} - \max_{a^\ast \in [m]}\left\{\ee{(a,a') \sim \tau}{u_2(a, a^\ast)}\right\} = o(1).
\]
Hence, as $T\to \infty$
\[
    d_{G,1}^{IVF}(\tau) = \max\left\{o(1),0\right\} = o(1),\,\, d_{G,2}^{IVF}(\tau) = \max\left\{o(1),0\right\} = o(1). 
    \]
Therefore, as $T \to \infty$,
$d_G(\tau) = \max_i\left\{\max\left\{d_{G,i}^{CE}(\tau),d_{G,i}^{IVF}(\tau)\right\}\right\} = o(1)$

\subsection{\texorpdfstring{Proof of \Cref{lemma:ivfconnection}}{Proof of Lemma \ref*{lemma:ivfconnection}}}\label{subsec:connectionproof}
\begin{lemma*}
    If the player has no swap regret and non-negative best-in-hindsight regret, then her learning is rationalizable and does not generate valuable information, where information is defined under the above framework.
\end{lemma*}
\paragraph{No swap regret implies fully exploiting information}
By definition of no swap regret, the player satisfies
\[
\max_{\sigma:[n] \to [n]} \frac 1T \sum_{t=1}^T \ee{a \sim p^t}{r^t_{\sigma(a)}} - \frac 1T \sum_{t=1}^T \ee{a \sim p^t}{r^t_a} = o(1),
\]
i.e.,
\[
\max_{\sigma:[n] \to [n]} \frac 1T \sum_{t=1}^T \ee{a \sim p^t}{r^t_{\sigma(a)}-r^t_a} = o(1),
\]
where $p^t$ is the player's strategy in round $t$.
When the player is in learning dynamics in which her rewards are determined by her opponent's action, let $u_1(a,a')$ denote her payoff $r^t_{a}$ of playing action $a$ when the opponent plays action $a_2^t=a'$.

Let $\tau^T$ be the empirical joint distribution of the player and her opponent's actions after $T$ rounds. No swap regret implies
\[
\max_{\sigma:[n] \to [n]}  \frac{1}{T}\sum_{t=1}^T\sum_{a'\in A_2}\ee{a \sim \tau^T(\cdot \mid a')}{u(\sigma(a),a')-u(a,a')} \ind{a_2^t=a'} =  o(1).
\]
Note that
\begin{align*}
    &\frac{1}{T}\sum_{t=1}^T\sum_{a'\in A_2}\ee{a \sim \tau^T(\cdot \mid a')}{u(\sigma(a),a')-u(a,a')} \ind{a_2^t=a'}\\
    = & \sum_{a'\in A_2}\ee{a \sim \tau^T(\cdot \mid a')}{u(\sigma(a),a')-u(a,a')} \sum_{t=1}^T\frac{\ind{a_2^t=a'}}{T}.
\end{align*}
Let $\tau^T_2$ be the distribution $\tau^T$ marginalized to the opponent's action. We have
\[
\sum_{t=1}^T\frac{\ind{a_2^t=a'}}{T}= \Pr_{a_2\sim \tau^T_2}[a_2 = a'].
\]
Therefore, no swap regret implies
\[
    \max_{\sigma:[n] \to [n]} \sum_{a'\in A_2}\ee{a \sim \tau^T(\cdot \mid a')}{u(\sigma(a),a')-u(a,a')} \Pr_{a_2\sim \tau_2^T}[a_2=a'] = o(1),
\]
and by definition of expectation over $\tau_2^T$,
\[
\max_{\sigma:[n] \to [n]} \ee{a' \sim \tau_2^T}{\ee{a \sim \tau^T(\cdot \mid a')}{u(\sigma(a),a')-u(a,a')}} = o(1),
\]
which is equivalent to, for all $\sigma$,
\[
\ee{a_2 \sim \tau_2^T}{\ee{a_1 \sim \tau^T(\cdot \mid a_2)}{u(a_1, a_2)} } \geq \ee{a_2 \sim \tau_2^T}{\ee{a_1 \sim \tau^T(\cdot \mid a_2)}{u(\sigma(a_1), a_2)} }-o(1).\]
Therefore, the player is approximately optimal conditional on the signal $\pi : \Omega \to \Delta(\Sigma)$ given by $\tau^T$ where $\pi(a_2) = \tau^T(\cdot\mid a_2)$.

\paragraph{Non-negative best-in-hindsight regret implies no valuable information}
By definition of non-negative best-in-hindsight regret, the player satisfies
\[
\max_{a^\ast \in [n]}\frac 1T \sum_{t=1}^T r^t_{a^\ast} - \frac 1T \sum_{t=1}^T \ee{a \sim p^t}{r^t_a} \geq -o(1),
\]
i.e.,
\[
\max_{a^\ast \in [n]} \frac 1T \sum_{t=1}^T r^t_{a^\ast}- \ee{a \sim p^t}{r^t_a} \geq -o(1),
\]
where $p^t$ is the strategy of the player in round $t$.
When the player is in learning dynamics in which her rewards are determined by her opponent's action, let $u_1(a,a')$ denote her payoff $r^t_{a}$ of playing action $a$ when the opponent plays action $a_2^t=a'$.

Let $\tau^T$ be the empirical joint distribution of the player and her opponent's actions after $T$ rounds. Non-negative best-in-hindsight regret implies
\[
\max_{a^\ast\in[n]}  \frac{1}{T}\sum_{t=1}^T\sum_{a'\in A_2}\ee{a \sim \tau^T(\cdot \mid a')}{u(a^\ast,a')-u(a,a')} \ind{a_2^t=a'} \geq -o(1).
\]
Note that
\begin{align*}
    &\frac{1}{T}\sum_{t=1}^T\sum_{a'\in A_2}\ee{a \sim \tau^T(\cdot \mid a')}{u(a^\ast,a')-u(a,a')} \ind{a_2^t=a'}\\
    = & \sum_{a'\in A_2}\ee{a \sim \tau^T(\cdot \mid a')}{u(a^\ast,a')-u(a,a')} \sum_{t=1}^T\frac{\ind{a_2^t=a'}}{T}.
\end{align*}
Let $\tau^T_2$ be the distribution $\tau^T$ marginalized to the opponent's action. We have
\[
\sum_{t=1}^T\frac{\ind{a_2^t=a'}}{T}= \Pr_{a_2\sim \tau^T_2}[a_2 = a'].
\]
Therefore, non-negative best-in-hindsight regret implies
\[
    \max_{a^\ast\in[n]} \sum_{a'\in A_2}\ee{a \sim \tau^T(\cdot \mid a')}{u(a^\ast,a')-u(a,a')} \Pr_{a_2 \sim \tau_2^T}[a_2=a'] \geq -o(1),
\]
and by definition of expectation over $\tau_2^T$,
\[
\max_{a^\ast \in [n]} \ee{a' \sim \tau_2^T}{\ee{a \sim \tau^T(\cdot \mid a')}{u(a^\ast,a')-u(a,a')}} \geq -o(1),
\]
which is equivalent to, for the maximizer $a^\ast$,
\[
\ee{a_2 \sim \tau_2^T}{\ee{a_1 \sim \tau^T(\cdot \mid a_2)}{u(a_1, a_2)} } \leq \ee{a_2 \sim \tau_2^T}{\ee{a_1 \sim \tau^T(\cdot \mid a_2)}{u(a^\ast, a_2)} }+o(1).\]
Since $a^\ast$ does not depend on the signal $\pi : \Omega \to \Delta(\Sigma)$ where $\pi(a_2) = \tau^T(\cdot\mid a_2)$, it still holds that
\[
\ee{a_2 \sim \tau_2^T}{\ee{a_1 \sim \tau^T(\cdot \mid a_2)}{u(a_1, a_2)}} \leq \ee{a_2 \sim \tau_2^T}{u(a^\ast, a_2)} + o(1).
\]
Therefore, the player is not (significantly) better off than she could have been without the additional information encoded in $\tau^T$.

\section{\texorpdfstring{Omitted Proofs in \Cref{sec:impos}}{Omitted Proofs in Section~\ref*{sec:impos}}}
\subsection{\texorpdfstring{Proof of \Cref{prop:smooth_example}}{Proof of Proposition~\ref*{prop:smooth_example}}}\label{subsec:pf_smooth_example}
\begin{proposition*}
The standard versions of the following algorithms are smooth learners:
    \begin{enumerate}
        \item Multiplicative Weights Update,
        \item Follow-the-Perturbed-Leader,
        \item Online Gradient Ascent (also known as Projected Gradient Ascent \citep{zinkevich2003online}).
    \end{enumerate}
    
    The corresponding Blum--Mansour reductions \citep{blum2007external} of these algorithms are also smooth learners.
\end{proposition*}

\paragraph{Multiplicative Weights Update}
Assume without loss of generality that $w_1=\frac 12$. Given any MWU algorithm with a learning rate $\eta$ such that $\eta=o(1)$ and $\eta T \to \infty$, let $\eps = 2/\eta T$. By the definition of MWU, at round $t$, the probability with which the row player plays the zero-reward action is
\[ w_t = \frac{1}{1+(1 + \eta)^\frac{2(t-1)}{\eta T}}. \]
Since $\eta>0$, $w_t$ is strictly decreasing. Furthermore,
\[ \int_a^{b+1} \frac{1}{1+(1 + \eta)^\frac{2(t-1)}{\eta T}}\,\mathrm dt \leq \sum_{t=a}^bw_t \leq \int_{a-1}^b \frac{1}{1+(1 + \eta)^\frac{2(t-1)}{\eta T}}\,\mathrm dt. \]
The integrals can be easily evaluated in closed form. For $(a, b)=(1, T/2)$ and $(T/2+1, T)$, the bounds are
\begin{align*}
    &\frac{T}{2} \underbracket{\left(1 - \frac{\eta}{\log(1 + \eta)} \log \left(\frac{1+(1 + \eta)^{\frac{T-2}{\eta T}}}{1+(1 + \eta)^{-\frac{2}{\eta T}}} \right) \right)}_{Q_1} \leq \sum_{t=1}^{T/2}w_t \leq \frac{T}{2} \underbracket{\left(1 + \frac{\eta}{\log(1 + \eta)} \log \left(\frac{2}{1+(1 + \eta)^\frac{1}{\eta}} \right) \right)}_{Q_2}, \\
    &\frac{T}{2} \underbracket{\left(1 - \frac{\eta}{\log(1 + \eta)} \log \left(\frac{1+(1 + \eta)^{\frac{2}{\eta}}}{1+(1 + \eta)^{\frac{1}{\eta}}} \right) \right)}_{Q_3} \leq \sum_{t = \frac T2+1}^Tw_t \leq \frac{T}{2} \underbracket{\left(1 + \frac 2T - \frac{\eta}{\log(1 + \eta)} \log \left(\frac{1+(1 + \eta)^{\frac{2T-2}{T \eta}}}{1+(1 + \eta)^{\frac{T-4}{T \eta}}} \right) \right)}_{Q_4}.
\end{align*}
Taking the limits of $Q_1$ through $Q_4$, we obtain
\[ \lim_{T \to \infty, \eta \to 0} Q_1 = \lim_{T \to \infty, \eta \to 0} Q_2=\frac 12-\frac 12\log \frac{1 + \e}{2}, \quad \lim_{T \to \infty, \eta \to 0} Q_3 = \lim_{T \to \infty, \eta \to 0} Q_4=\frac 12-\frac 12\log \frac{1 + \e^2}{1 + \e}. \]
It follows that $\sum_{t=1}^{T/2}w_t = \left(\frac 12-\frac 12\log \frac{1 + \e}{2} \right)T+o(T)$ and $\sum_{t=T/2+1}^Tw_t = \left(\frac 12-\frac 12\log \frac{1 + \e^2}{1 + \e} \right)T+o(T)$. Therefore, $\sum_{t=1}^{T/2}w_t - \sum_{t=T/2+1}^Tw_t = \Omega(T)$.

\paragraph{Follow-the-Perturbed-Leader}
Assume without loss of generality that $w_1=\frac 12$. Given any FTPL algorithm with learning rate $\eta$, let $\eps=\frac{1}{\eta T}$. Let $X, Y$ be two i.i.d. random variables with exponential distribution $\mathrm{Exp}(\eta)$. By the definition of FTPL, at round $t$, the probability with which the row player plays the zero-reward action is
\[ w_t = \Pr[X \geq \eps(t-1)+Y] = \Pr[X-Y \geq \eps(t-1)] = \frac{1}{2} \exp(-\eta \eps(t-1)) = \frac{1}{2} \exp(-(t-1)/T). \]
The third equality follows from the fact that the difference between two i.i.d. exponentially distributed random variables follows the Laplace distribution.

The sums of the weights can now be evaluated in closed form.
\[ \sum_{t=1}^{T/2}w_t = \frac 12 \cdot \frac{1 - \e^{-\frac 12}}{1 - \e^{-\frac 1T}}, \quad \sum_{t=T/2+1}^Tw_t = \frac 12 \cdot \frac{\e^{-\frac 12}-\e^{-1}}{1 - \e^{-\frac 1T}} = \e^{-\frac 12} \sum_{t=1}^{T/2}w_t. \]
So it suffices to show $\sum_{t=1}^{T/2}w_t = \Omega(T)$. But $1 - \e^{-\frac 1T} \leq \frac{1}{T}$. Therefore, $\sum_{t=1}^{T/2}w_t \geq \frac{T}{2} \left(1 - \e^{-\frac 12} \right)=\Omega(T)$ as desired.

\paragraph{Online/Projected Gradient Ascent}

Consider online gradient ascent with weakly decreasing step size $\eta_t = \Theta(1/\sqrt t)$.

Let $x_t=(w_t, 1-w_t)$, where $w_t$ is the probability of the zero-reward arm in round $t$. As in \Cref{def:smooth}, choose the zero-reward arm so that $w_1 = \Omega(1)$. Let
\[ S_1 = \sum_{s=1}^{T/2} \eta_s, \quad \eps = \frac{w_1}{S_1}. \]
We have $S_1 = \Theta(\sqrt T) $ as $\eta_t = \Theta(1/\sqrt t)$, and hence $\eps = \Theta(1/\sqrt T)=o(1)$.

The reward vector is $(0, \eps)$, so projected gradient ascent updates by
\[ x_{t+1} = \Pi_{\Delta_2}(x_t + \eta_t(0, \eps)). \]
In the two-dimensional simplex, Euclidean projection subtracts the same amount from both coordinates, unless one coordinate is clipped to zero. Therefore, the first coordinate satisfies
\begin{equation}\label{eqn:pga_update}
    w_{t+1} = \left(w_t - \frac{\eps \eta_t}{2} \right)_+.
\end{equation}
From \Cref{eqn:pga_update}, $w_t$ is weakly decreasing.

It remains to prove the smooth-learner property. Iterating \Cref{eqn:pga_update} and using $(x_+-y)_+=(x-y)_+$ for any $y \geq 0$, we have
\begin{equation}\label{eqn:pga_iteration}
    w_t = \left(w_1 - \frac \eps 2 \sum_{s=1}^{t-1} \eta_s \right)_+.
\end{equation}
Note that for every $t \leq T/2$, the quantity subtracted in \Cref{eqn:pga_iteration} is at most $\eps S_1/2 = \frac{w_1}{2}$. Therefore $w_t \geq w_1/2 = \Omega(1)$.

Now fix any $t \leq T/2$. Applying \Cref{eqn:pga_update} from round $t$ to round $t+T/2$ gives
\begin{equation}\label{eqn:oga_diff}
    w_t-w_{t+T/2} = \min \left\{w_t, \frac \eps 2 \sum_{s=t}^{t+T/2-1} \eta_s \right\}.
\end{equation}
The first term inside the minimum is $\Omega(1)$ because, as shown above, $w_t \geq w_1/2$ for every $t \leq T/2$. For the second term, since $\eta_t$ is weakly decreasing,
\[ \frac \eps 2 \sum_{s=t}^{t+T/2-1} \eta_s \geq \frac \eps 2 \sum_{s=T/2}^{T-1} \eta_s = \frac{w_1}{2} \cdot \frac{\sum_{s=T/2}^{T-1} \eta_s}{\sum_{s=1}^{T/2} \eta_s} = \Omega(1). \]
Here the first equality follows from the definition of $\eps$, and the last equality again uses $\eta_t = \Theta(1/\sqrt t)$. Therefore $w_t-w_{t+T/2} = \Omega(1)$ for every $t \leq T/2$. Summing over $t=1, \dotsc, T/2$ we obtain $\sum_{t=1}^{T/2}w_t - \sum_{t=T/2+1}^Tw_t = \sum_{t=1}^{T/2}(w_t-w_{t+T/2}) = \Omega(T)$. This finishes the proof.

We also note that online gradient ascent is not a mean-based algorithm \citep[Appendix B]{ahunbay2025semicoarse}.

\paragraph{Blum--Mansour reductions}
Recall that in the two-arm case, the Blum--Mansour reduction \citep{blum2007external} maintains two copies of no-best-in-hindsight-regret algorithms. Let the first copy's output distribution be $(x_t, 1-x_t)$ and let the second copy's be $(y_t, 1-y_t)$ in round $t$. The Blum--Mansour reduction forms the transition matrix with these two vectors as rows and solves for the stationary distribution
\[ \begin{pmatrix}
    x_t & 1-x_t \\ y_t & 1-y_t
\end{pmatrix}^\top \begin{pmatrix}
    w_t \\ 1-w_t
\end{pmatrix} = \begin{pmatrix}
    w_t \\ 1-w_t
\end{pmatrix}, \]
which gives $w_t = \frac{y_t}{1-x_t+y_t}$. After the reward vector is observed, the first copy receives the feedback of $w_t(0, \eps)$ and the second copy receives $(1-w_t)(0, \eps)$.

We first make some common observations for MWU, FTPL, and OGA. Again, assume without loss of generality that they all start with the uniform distribution over the two arms. Since each copy is always given a reward vector where the second arm is better, $x_t, y_t$ are weakly decreasing in $t$ and both $x_t, y_t \leq \frac 12$. Since the map $(x, y) \mapsto y/(1-x+y)$ is increasing in both coordinates in $[0,1] \times [0, 1]$, the sequence $\{w_t\}_{t=1}^T$ is weakly decreasing. Moreover, by $x_t, y_t \leq \frac 12$ we also get $w_t \leq \frac 12$ for all $t$.

For the separation between $\sum_{t=1}^{T/2}w_t$ and $\sum_{t=T/2+1}^Tw_t$, we prove a common reduction step. Suppose that for some constant $\delta>0$, the second copy satisfies
\begin{equation}\label{eqn:bm_smooth_redux}
    y_t-y_{t+T/2} \geq \delta,
\end{equation}
for every $t \leq T/2$. Then, since $x_{t+T/2} \leq x_t$ and the stationary-distribution map is increasing in both coordinates, we have
\begin{align*}
    w_t-w_{t+T/2} &\geq \frac{y_t}{1-x_t+y_t} - \frac{y_t-\delta}{1-x_t+y_t-\delta} = \frac{\delta(1-x_t)}{(1-x_t+y_t)(1-x_t+y_t - \delta)}
    \\&\geq \frac{\delta(1-x_t)}{(1-x_t+y_t)^2} \geq \delta \frac{1 - \frac 12}{\left(\frac 32 \right)^2} = \frac{2\delta}{9} \geq \frac \delta 8.
\end{align*}
Here for the second inequality we use $\delta>0$ and for the third inequality we use $x_t, y_t \in \left[0, \frac 12 \right]$. Summing over $t=1, 2, \dotsc, T/2$ we obtain $\sum_{t=1}^{T/2}w_t - \sum_{t=T/2+1}^Tw_t=\Omega(T)$. Therefore, it suffices to verify \Cref{eqn:bm_smooth_redux} for all three subroutines.

For MWU, let $B_t = \sum_{s=1}^{t-1}(1-w_s)$. With the same choice as in the MWU part of the proof, $\eps=2/(\eta T)$, the probability that the second copy recommends the zero-reward arm is $y_t=1/(1+(1+\eta)^{\eps B_t})$. Since $w_s \leq \frac 12$, we have $B_{t+T/2}-B_t \geq T/4$ and for $t \leq T/2$ we have $B_t \leq T/2$. Now consider $a_t=(1+\eta)^{\eps B_t}$. Then $1 \leq a_t \leq (1+\eta)^{\frac 1 \eta} \leq \mathrm e$ where the second inequality follows from $B_t \leq T/2$ and $\eps=2/(\eta T)$. Furthermore,
\begin{align*}
    y_t-y_{t+T/2} &= \frac{1}{1+a_t} - \frac{1}{1+a_t(1 + \eta)^{\eps(B_{t+T/2}-B_t)}} \geq \frac{1}{1+a_t} - \frac{1}{1+a_t \rho}
    \\&= \frac{a_t(\rho-1)}{(1+a_t)(1+a_t \rho)},
\end{align*}
where $\rho:=(1+\eta)^{\frac{1}{2\eta}}$. Here the first inequality follows from $B_{t+T/2}-B_t \geq T/4$ and $\eps=2/(\eta T)$. Observe that $\rho \to \mathrm e^{\frac 12}$ from below as $T \to \infty$ under the same choice of learning rate as above. Therefore, for all sufficiently large $T$,
\[ \frac{a_t(\rho-1)}{(1+a_t)(1+a_t \rho)} \geq \frac{\mathrm e^{\frac 14}-1}{(1 + \mathrm e)(1+\mathrm e^{\frac 32})}>0. \]
Here for the first inequality we use $1 \leq a_t \leq \mathrm e$ and for sufficiently large $T$, $\mathrm e^{\frac 14} \leq \rho \leq \mathrm e^{\frac 12}$. Since the right-hand side is bounded below by a positive constant, \Cref{eqn:bm_smooth_redux} holds and the Blum--Mansour reduction of MWU is a smooth learner.

For FTPL, with the same definition of $B_t$ and the same choice as in the FTPL part of the proof, $\eps=1/(\eta T)$, the second copy satisfies $y_t = \frac 12 \exp(-\eta \eps B_t)$. The same bounds $B_t \leq T/2$ and $B_{t+T/2}-B_t \geq T/4$ now give
\[ y_t-y_{t+T/2} = \frac 12 \exp(-\eta \eps B_t) \left(1 - \exp(\eta \eps(B_t-B_{t+T/2})) \right) \geq \frac 12 \mathrm e^{-\frac 12} \left(1 - \mathrm e^{-\frac 14} \right)>0, \]
and \Cref{eqn:bm_smooth_redux} holds and the Blum--Mansour reduction of FTPL is a smooth learner.

For OGA, let the step sizes satisfy $\eta_t = \Theta(1/\sqrt t)$ as in the preceding proof. Let $S = \sum_{s=1}^{T/2}\eta_s$ and choose $\eps=1/(4S)$. For the second copy, $y_t$ evolves with
\[ y_{t+1}=\left(y_t - \eps \eta_t \frac{1-w_t}{2} \right)_+. \]
Therefore, by an argument similar to that used for \Cref{eqn:oga_diff}, we get
\[ y_t-y_{t+T/2} = \min \left\{y_t, \sum_{s=t}^{t+T/2-1} \frac{\eps \eta_s(1-w_s)}{2} \right\}. \]
Before $T/2$, the total possible decrease in $y_t$ is at most $\eps \frac S2 = \frac 18$, so $y_t \geq \frac 38$ for every $t \leq T/2$. On the other hand, from $t$ to $t+T/2$, the total decrease
\[ \sum_{s=t}^{t+T/2-1} \frac{\eps \eta_s(1-w_s)}{2} \geq \frac \eps 4 \sum_{s=t}^{t+T/2-1} \eta_s. \]

Finally, we use $\eta_s = \Theta(1/\sqrt s)$. There are constants $m, M>0$ such that for all sufficiently large $T$, $m/\sqrt s \leq \eta_s \leq M/\sqrt s$. So
\[ \sum_{s=t}^{t+T/2-1} \eta_s \geq \sum_{s=T/2}^{T-1}\eta_s \geq m \sum_{s=T/2}^{T-1}\frac{1}{\sqrt s} \geq c_0 \sum_{s=1}^{\frac T/2}\eta_s=c_0S, \]
for some constant $c_0>0$. Here the first inequality follows since $\eta_s$ is decreasing. The final inequality follows because $\sum_{s=T/2}^{T-1} \frac{1}{\sqrt s} \geq \int_{T/2}^Tx^{-\frac 12}\,\mathrm dx=2 \sqrt T-2 \sqrt{T/2}$, $\sum_{s=1}^{T/2}\frac{1}{\sqrt s} \leq 1+\int_1^{T/2}x^{-\frac 12}\,\mathrm dx \leq 2 \sqrt{T/2}$, and $\frac{m(2 \sqrt T-2 \sqrt{T/2})}{M(2 \sqrt{T/2})} = \frac{m}{M}(\sqrt 2-1)$. By the definition $\eps=1/(4S)$ we have $\frac \eps 4 \sum_{s=t}^{t+T/2-1}\eta_s \geq \frac \eps 4 \cdot c_0S = \frac{c_0}{16}$. Thus we conclude
\[ y_t-y_{t+T/2} \geq \min \left\{\frac 38, \frac{c_0}{16} \right\}, \]
and \Cref{eqn:bm_smooth_redux} holds and the Blum--Mansour reduction of OGA is a smooth learner. This concludes the proof of \Cref{prop:smooth_example}.

\subsection{\texorpdfstring{Proof of \Cref{thm:impos-ivfe}}{Proof of Theorem~\ref*{thm:impos-ivfe}}}\label{subsec:proof-impos-ivfe}
\begin{theorem*}
Fix any learning algorithm $\mathcal L$ with no swap regret against any smooth learner $\mathcal A$. There exists a constant $c>0$ (depending on the smooth-learning property of $\mathcal A$) and a sequence of games $\{G_T\}_{T=1}^\infty$ such that, when the players play $G_T$ for $T$ rounds, the empirical history of play $\tau_T$ remains an asymptotically constant distance from an IVFE:
    \[ \liminf_{T \to \infty} d_{G_T}(\tau_T) \geq c. \]
\end{theorem*}

\paragraph{Construction}
Given any time horizon $2T$ and any smooth learner, we construct a game $G_{2T}$ such that the players cannot converge to an IVFE in the game.

In the construction of $G_{2T}$, the row player uses a smooth learner. The players play $2T$ rounds of the fixed underlying game shown in \Cref{table:impossibility-ivfe-game}. The parameters are to be set in a moment.
\begin{table}[htbp]
    \centering
    \begin{tabular}{|c|c|c|c|}
        \hline
        & $L$    &  $M$    & $R$  \\ \hline
         $U$ & $0, U_L$    & $0, U_M$    & $0, 0$    \\ \hline
         $D$ & $\eps, D_L$ & $\eps, D_M$ & $\eps, 0$ \\ \hline
    \end{tabular}
    \caption{Payoff matrix of the game $G_{2T}$}
    \label{table:impossibility-ivfe-game}
\end{table}
The row player's utility depends only on her own action, and hence the row player's behavior is exactly the behavior of the smooth learner in the two-arm experiment with rewards 0 and $\eps$. Let $w_t$ be the probability that the row player plays $U$ in round $t$. Define
\[ P = \sum_{t=1}^Tw_t, \quad Q = \sum_{t=T+1}^{2T}w_t. \]
By smoothness, $P-Q = \Omega(T)$.

Let
\[ \overline w = \frac{1}{2T} \sum_{t=1}^{2T}w_t, \quad A = \sum_{t:w_t>\overline w}(w_t - \overline w), \quad \beta = \frac{A}{8T}. \]
We set the payoffs to be
\[ U_L=1 - \overline w - \beta, \quad U_M = \overline w-1, \]
and
\[ D_L=-\overline w - \beta, \quad D_M = \overline w. \]
Note that $0 \leq w_t \leq 1$ for all $t$, therefore $\overline w \in [0, 1]$, and all four payoffs $U_L, D_L, U_M, D_M$ are $O(1)$. Moreover, observe that $A \leq 2T(1 - \overline w)$, so $\beta = \frac{A}{8T} \leq \frac{1 - \overline w}{4} \leq \frac 14$ and $\overline w+ \beta \leq \frac 34 \overline w+\frac 14 \leq 1$.

\paragraph{Analysis}

We show that the distance to an IVFE is at least $\eta=\beta^2$. Together with $A=\Omega(T)$ shown in \Cref{lemma:period1_LM} and $\beta = \frac{A}{8T}$, we conclude that there exists a constant $c>0$ such that the limit inferior is at least $c$.

Assume for the sake of contradiction that the distance is less than $\eta = \beta^2$. Then there exist infinitely many time horizons such that the magnitude of negative best-in-hindsight regret \emph{in the constructed game} is less than $2\eta T+O(TR(T))$ for some rate $R(T)=o(1)$. Take such a subsequence and re-index it from $1, 2, \dotsc$.

Let $R'(T)=o(1)$ be a rate such that over the $2T$-round horizon, the column player's swap regret is at most $O(TR'(T))$. It is without loss of generality to assume $R'(T)=O(R(T))$ as we can always take $\max\{R(T), R'(T)\}$ as the general rate and this does not change any assumption since both rates are $o(1)$. Therefore, we replace $R(T), R'(T)$ by their maximum and use the same notation $R(T)$ afterwards.

Throughout the remaining proof, $h(T)$ denotes the error term from the smoothness property for the $2T$-round run. (Equivalently one may read it as $h(2T)$, but we suppress this notational distinction.)

Let $a_t \in \{L, M, R\}$ be the action chosen by the column player in round $t$. Let $G_i^{(t)}$ denote the column player's reward from playing action $i \in \{L, M, R\}$ in round $t$. Since $G_R^{(t)}=0$ for every $t$, the column player's actual payoff over any set of rounds is the sum of only the contributions from the rounds in which she plays $L$ or $M$.

By smoothness, there exists a set $H \subseteq \{1, 2, \dotsc, 2T\}$ of size at least $2T(1-h(T))$ such that $\{w_t:t \in H\}$ is weakly decreasing in time. Consider
\[ H^+ = \{t \in H:w_t>\overline w+\beta\}, \quad \tau = \max H^+. \]
Later in \Cref{lemma:period1_LM} we show that $H^+$ is non-empty and $\tau$ is well-defined for all sufficiently large $T$. We call rounds $1, 2, \dotsc, \tau$ \emph{period one}, and rounds $\tau+1, \tau+2, \dotsc, 2T$ \emph{period two}.

Now, the intuition of our proof is similar to the simple case in the main text. To have no swap regret, the column player must get a significant amount of reward from action $M$ in period two. It follows that to have non-negative best-in-hindsight regret, the player must have low reward in period one (\Cref{lemma:low_period1}). However, this observation contradicts the fact that the player's algorithm has no swap regret against any smooth learner and hence the reward in period one should be high (\Cref{lemma:high_period1}). Formally, we show
\begin{lemma}\label{lemma:high_period1}
    We have
    \[ \ee{}{\sum_{t=1}^\tau \left(\ind{a_t=L}G_L^{(t)} + \ind{a_t=M}G_M^{(t)} \right)} \geq \frac 34A-o(T)= \Omega(T). \]
\end{lemma}
\begin{lemma}\label{lemma:low_period1}
    We have
    \[ \ee{}{\sum_{t=1}^\tau \left(\ind{a_t=L}G_L^{(t)} + \ind{a_t=M}G_M^{(t)} \right)} \leq \frac{5\eta}{2 \beta}T+ o(T). \]
\end{lemma}

The contradiction immediately follows from \Cref{lemma:high_period1,lemma:low_period1} because, if $\eta \leq \frac{A}{2T}$, then
\[ \frac 34 \cdot (8 \beta T)-o(T) \leq \ee{}{\sum_{t=1}^\tau \left(\ind{a_t=L}G_L^{(t)} + \ind{a_t=M}G_M^{(t)} \right)} \leq \frac{5 \beta^2T}{2 \beta}+o(T), \]
which implies $6\beta T \leq \frac 52 \beta T$, which is a contradiction.

\subsection{\texorpdfstring{Proof of \Cref{lemma:high_period1}}{Proof of Lemma~\ref*{lemma:high_period1}}}\label{subsec:pf_lemma_high_period1}
\begin{lemma*}
    We have
    \[ \ee{}{\sum_{t=1}^\tau \left(\ind{a_t=L}G_L^{(t)} + \ind{a_t=M}G_M^{(t)} \right)} \geq \frac 34A-o(T)= \Omega(T). \]
\end{lemma*}

We first make the following observation on the payoff structure.
\begin{lemma}\label{lemma:er}
    Let $G_i^{(t)}\,(i \in \{L, M, R\})$ be the random variable that denotes the reward from playing $i$ in round $t$ for the column player. Then
    \begin{align*}
        \ee{}{G_L^{(t)}}=w_t-\overline w - \beta, \quad & \ee{}{\sum_{t=1}^{2T}G_L^{(t)}}=-\frac A4, \\
        \ee{}{G_M^{(t)}} = \overline w-w_t, \quad & \ee{}{\sum_{t=1}^{2T}G_M^{(t)}}=0.
    \end{align*}
\end{lemma}
The proof of \Cref{lemma:er} is algebraic, and the reader may skip it without affecting the understanding of the main argument.
\begin{proof}
    We verify that
    \begin{align*}
        \ee{}{G_L^{(t)}}&=w_t(1 - \overline w - \beta)+(1-w_t)(-\overline w - \beta)=w_t+1 \cdot (-\overline w-\beta)
        \\ \ee{}{G_M^{(t)}}&=w_t(\overline w-1)+(1-w_t) \overline w = \overline w-w_t,
    \end{align*}
    and that
    \begin{align*}
        \ee{}{\sum_{t=1}^{2T}G_L^{(t)}} &= \sum_{t=1}^{2T}(w_t-\overline w)-2T \beta=-2T \beta=-\frac A4,
        \\\ee{}{\sum_{t=1}^{2T}G_M^{(t)}} &= \sum_{t=1}^{2T}(w_t-\overline w)=0.
    \end{align*}
    Here $\sum_{t=1}^{2T}(w_t-\overline w)=0$ since $\overline w$ is the average of the values $w_t$.
\end{proof}

Before presenting the proof, we also need the following lemma.
\begin{lemma}\label{lemma:period1_LM}
    The set $H^+$ is non-empty and $\tau$ is well-defined for all sufficiently large $T$. Moreover, we have
    \[ \ee{}{\sum_{t=1}^\tau G_L^{(t)}} \geq \frac 34A-o(T)= \Omega(T), \quad \ee{}{G_M^{(t)}} \leq -\beta, \]
    for every $t \leq \tau$ and $t \in H$.
\end{lemma}
The proof of \Cref{lemma:period1_LM} can be found in \Cref{subsec:proof_period1_LM}.

Consider a counterfactual row player that plays action $U$ with probability $w_t$ for $t \leq \tau$ and with probability $\overline w+ \beta$ for $t > \tau$. Recall that $\overline w+\beta \leq 1$ as we have shown at the end of the construction paragraph. We also claim:
\begin{claim}\label{claim:counterfactual_smooth}
    This counterfactual row player is a smooth learner.
\end{claim}
The proof of \Cref{claim:counterfactual_smooth} can be found in \Cref{subsec:proof_claim_smooth}.

The column player's algorithm has no swap regret against any smooth learner, so it has swap regret at most $O(TR(T))$ even against this counterfactual smooth learner. Moreover, the first $\tau$ histories in the counterfactual sequence and in the actual sequence are identical, so the column player's first $\tau$ strategies are the same in the two processes.

In the counterfactual sequence, consider the constant swap that maps every action to $L$.

For $t > \tau$, by a similar calculation as in \Cref{lemma:er}, we have $\ee{}{G_L^{(t)}}=(1-\overline w-\beta)(\overline w+\beta)-(\overline w+\beta)(1-\overline w-\beta)=0$ and $\ee{}{G_M^{(t)}}=(\overline w-1)(\overline w+\beta)+\overline w(1 - \overline w-\beta)=-\beta \leq 0$. Therefore, the contribution to swap regret in the second period is non-negative, and furthermore the contribution to swap regret in the first period is $O(TR(T))$. Finally, since the column player's first $\tau$ strategies are the same in the two processes, we have
\begin{align*}
    \ee{}{\sum_{t=1}^\tau \left(\ind{a_t=L}G_L^{(t)} + \ind{a_t=M}G_M^{(t)} \right)} &\geq \ee{}{\sum_{t=1}^\tau G_L^{(t)}}-O(TR(T))
    \\& \geq \frac 34A-o(T)-O(TR(T))=\Omega(T),
\end{align*}
where the last equality follows from \Cref{lemma:period1_LM}. This proves the lemma.

\subsection{\texorpdfstring{Proof of \Cref{lemma:low_period1}}{Proof of Lemma~\ref*{lemma:low_period1}}}
\begin{lemma*}
    We have
    \[ \ee{}{\sum_{t=1}^\tau \left(\ind{a_t=L}G_L^{(t)} + \ind{a_t=M}G_M^{(t)} \right)} \leq \frac{5\eta}{2 \beta}T+ o(T). \]
\end{lemma*}

Before presenting the proof, we need the following lemma.
\begin{lemma}\label{lemma:nl_bound}
    We have
    \begin{equation}\label{eqn:lem_nl_bound}
        \sum_{t=1}^{2T} \Pr[a_t=L] \leq \frac{2\eta}{\beta}T+o(T).
    \end{equation}
\end{lemma}
The proof of \Cref{lemma:nl_bound} is deferred to \Cref{subsec:nl_bound_pf}.

Now we bound the $L$-contribution and the $M$-contribution separately.

First consider the $L$-contribution. 
\begin{align}
    \ee{}{\sum_{t=1}^\tau \ind{a_t=L}G_L^{(t)}} &\leq \ee{}{\sum_{t=1}^\tau \ind{a_t=L}} \max_t\|G_L^{(t)}\|_\infty \nonumber
    \\&\leq \ee{}{\sum_{t=1}^{2T} \ind{a_t=L}} \max_t\|G_L^{(t)}\|_\infty
    \leq \frac{5\eta}{2 \beta}T+o(T). \label{eqn:tau_L}
\end{align}
Here the second inequality follows because the expected number of times the player plays action $L$ in rounds $1, 2, \dotsc, \tau$ is bounded by the expected number over the whole horizon. The final inequality follows because by \Cref{lemma:nl_bound}, the expected total number of times the learner plays action $L$ over the whole time-horizon is $\frac{2c}{\beta}T+o(T)$ and $\|G_L^{(t)}|_\infty \leq \frac 54$ because $\overline w \in [0, 1]$ and $\beta \leq \frac 14$.

Now consider the $M$-contribution. Recall that $\ee{}{G_M^{(t)}}=\overline w-w_t$ by \Cref{lemma:er}. Write
\begin{align}
    \ee{}{\sum_{t=1}^\tau \ind{a_t=M}G_M^{(t)}} &= \ee{}{\sum_{\substack{1 \leq t \leq \tau \\ t \in H^+}} \ind{a_t=M}G_M^{(t)}} + \ee{}{\sum_{\substack{1 \leq t \leq \tau \\ t \notin H^+}} \ind{a_t=M}G_M^{(t)}} \nonumber
    \\&\leq \ee{}{\sum_{\substack{1 \leq t \leq \tau \\ t \in H^+}} \ind{a_t=M}G_M^{(t)}}+|\{t \leq \tau:t \notin H\}| \cdot \max_t\|G_M^{(t)}\|_\infty \nonumber
    \\& = \sum_{\substack{1 \leq t \leq \tau \\ t \in H^+}} \Pr[a_t=M] \ee{}{G_M^{(t)}}+|\{t \leq \tau:t \notin H\}| \cdot \max_t\|G_M^{(t)}\|_\infty \nonumber
    \\&\leq 0+|\{t \leq \tau:t \notin H\}| \cdot \max_t\|G_M^{(t)}\|_\infty \nonumber
    \\&\leq (2T-|H|) \cdot \max_t\|G_M^{(t)}\|_\infty \nonumber
    \\&\leq 2Th(T) \cdot O(1)=o(T). \label{eqn:tau_M}
\end{align}
Here the first inequality follows from the proof of \Cref{lemma:period1_LM}, in which we show every $t \in H$ and $t \leq \tau$ belongs to $H^+$. The second equality follows from linearity of expectation, and the fact that $a_t$ and $G_M^{(t)}$ are independent conditional on the past since $a_t$ is chosen using only the history before round $t$, while $G_M^{(t)}$ is determined by the row player's current-round randomization. The second inequality follows by the definition of $H^+$ and $\ee{}{G_M^{(t)}}=\overline w-w_t$. The final inequality follows from the assumption of smoothness $|H| \geq 2T(1-h(T))$ and the absolute values of all payoffs in the constructed game are $O(1)$, as explained above. The lemma now follows by adding \Cref{eqn:tau_L,eqn:tau_M}.

\subsection{\texorpdfstring{Proof of \Cref{lemma:period1_LM}}{Proof of Lemma~\ref*{lemma:period1_LM}}}\label{subsec:proof_period1_LM}
\begin{lemma*}
    The set $H^+$ is non-empty and $\tau$ is well-defined for all sufficiently large $T$. Moreover, we have
    \[ \ee{}{\sum_{t=1}^\tau G_L^{(t)}} \geq \frac 34A-o(T)= \Omega(T), \quad \ee{}{G_M^{(t)}} \leq -\beta, \]
    for every $t \leq \tau$ and $t \in H$.
\end{lemma*}

We first make a few observations on the properties of the values $w_t$. Consider $\sum_{t=1}^T(w_t - \overline w)$. Write
\begin{equation}\label{eqn:first_half_sum}
    \sum_{t=1}^T(w_t - \overline w)=P-T \overline w=P - T \cdot \frac{P+Q}{2T} = \frac{P-Q}{2} = \Omega(T).
\end{equation}
Here the first equality follows since $P = \sum_{t=1}^Tw_t$. The second equality follows since $\overline w = \frac{P+Q}{2T}$ where $Q = \sum_{t=T+1}^{2T}w_t$. The final equality follows from the assumption of smoothness.

Recall that $A = \sum_{t:w_t - \overline w>0}(w_t - \overline w)$. Write
\[ \sum_{t=1}^T(w_t - \overline w) \leq \sum_{\substack{1 \leq t \leq T \\ w_t - \overline w>0}}(w_t - \overline w) \leq \sum_{\substack{1 \leq t \leq 2T \\ w_t - \overline w>0}}(w_t - \overline w)=A. \]
The first inequality follows because the first-half sum $\sum_{t=1}^T(w_t - \overline w)$ may contain both positive and negative terms, and removing the negative terms can only increase the sum. The second inequality follows by adding the positive contributions in $T+1, \dotsc, 2T$. Together with \Cref{eqn:first_half_sum} we conclude $A = \Omega(T)$. As a fact that will be used in other proofs, we also have $\beta = \frac{A}{8T} = \Theta(1)$. Indeed, from the fact that $A = \Omega(T)$, which we have just shown, there exists a constant $C>0$ such that for all sufficiently large $T$, $A \geq CT$. Therefore, for a constant $C' = \frac C8>0$ and all sufficiently large $T$
\begin{equation}\label{eqn:beta_theta1}
    C' \leq \beta \leq \frac 14
\end{equation}
as we also have $A \leq 2T$.

Now observe that
\begin{align}
    \sum_{t \in H^+}(w_t - \overline w - \beta) &\geq \sum_{t:w_t>\overline w + \beta}(w_t-\overline w-\beta)-2Th(T) \nonumber \\& \geq A-2T \beta-2Th(T) = \frac 34A-2Th(T). \label{eqn:H_plus_bound}
\end{align}
The first inequality follows because there are at most $2Th(T)$ rounds outside $H$ by the smoothness assumption, and each $w_t - \overline w$ contributes at most one. The second inequality follows from the fact that $\max\{x-\beta, 0\} \geq \max\{x, 0\}-\beta$ and the definition of $A$. The first equality follows from the definition of $
\beta$. In particular, $\sum_{t \in H^+}(w_t - \overline w - \beta) \geq \Omega(T)$ as $A = \Omega(T)$ and $h(T)=o(1)$. Therefore, $H^+$ is non-empty and $\tau = \max H^+$ is well-defined.

Now we show the property regarding $G_M^{(t)}$. By the smoothness assumption, if $t \leq \tau$ and $t \in H$ then $w_t \geq w_\tau$. Since $\tau \in H^+$, we have $w_\tau>\overline w+\beta$. Therefore $w_t \geq w_\tau>\overline w+\beta$ and $\ee{}{G_M^{(t)}}=\overline w-w_t \leq -\beta$. (In particular $w_t>\overline w+\beta$ and $t \in H^+$, so every $t \in H$ and $t \leq \tau$ belongs to $H^+$.) This proves the second claim.

It remains to lower bound the $L$-payoff before $\tau$.

Recall that $\ee{}{G_L^{(t)}}=w_t - \overline w - \beta$ by \Cref{lemma:er}. We have
\begin{align*}
    \ee{}{\sum_{t=1}^\tau G_L^{(t)}} &= \sum_{t=1}^\tau(w_t - \overline w - \beta) \\&\geq \sum_{t \in H^+}(w_t - \overline w - \beta)-2Th(T)
    \\& \geq \frac 34A-4Th(T)=\Omega(T).
\end{align*}
For the first inequality, note we have just shown that each $t \in H$ and $t \leq \tau$ belongs to $H^+$. So the only potential negative terms in the sum come from rounds not in $H$, of which there are at most $2Th(T)$ rounds in the whole horizon. The second inequality follows from \Cref{eqn:H_plus_bound}. The lemma follows by noting that $A = \Omega(T)$ and $h(T)=o(1)$.

\subsection{\texorpdfstring{Proof of \Cref{claim:counterfactual_smooth}}{Proof of Claim~\ref*{claim:counterfactual_smooth}}}\label{subsec:proof_claim_smooth}
\begin{claim*}
    The counterfactual row player constructed in \Cref{subsec:pf_lemma_high_period1} is a smooth learner.
\end{claim*}

Let
\[ \tilde w_t = \begin{cases} w_t & (t \leq \tau), \\ \overline w+\beta & (t>\tau). \end{cases} \]

First, the sequence is weakly decreasing over $(H \cap \{1, 2, \dotsc, \tau\}) \cup \{\tau+1, \tau+2, \dotsc, 2T\}$ because
\begin{itemize}
    \item it is decreasing over $H \cap \{1, 2, \dotsc, \tau\}$ by definition, and
    \item for every $t \leq \tau$ and $t \in H$ we have $w_t \geq w_\tau>\overline w+\beta$ (proven in the proof of \Cref{lemma:period1_LM}).
\end{itemize}
This set has size at least $2T-2Th(T)$ by the definition of $H$.

It remains to prove the difference between the first-half and second-half sums. We first observe that $\tau$ cannot be too large. In fact, since $\sum_{t=1}^{2T}(w_t-\overline w)=0$ we have
\[ \sum_{t:w_t\leq \overline w}(\overline w-w_t)=A. \]
Because each summand is at most one, there are at least $A$ rounds with $w_t \leq \overline w$ and at least $A-2Th(T)$ of them belong to $H$. Furthermore, every such round occurs after $\tau$, since every $t \in H$ with $t \leq \tau$ satisfies $w_t>\overline w+\beta$. Thus
\begin{equation}\label{eqn:small_tau}
    2T-\tau \geq A-2Th(T), \quad \text{or equivalently} \quad \tau \leq 2T-A+2Th(T).
\end{equation}
In particular, $|H^+| \leq 2T-A+2Th(T)$ because $\tau = \max H^+$.

We then proceed in two cases.

\paragraph{$\tau \leq T$}
Then we have
\begin{align*}
    \sum_{t=1}^T\tilde w_t-\sum_{t=T+1}^{2T}\tilde w_t &= \sum_{t=1}^\tau w_t+\sum_{t=\tau+1}^T(\overline w+\beta)-\sum_{t=T+1}^{2T}(\overline w+\beta)
    \\&=\sum_{t=1}^\tau(w_t-\overline w-\beta)=\Omega(T).
\end{align*}
Here the final equality follows by a derivation similar to that of \Cref{eqn:H_plus_bound}.

\paragraph{$\tau>T$}
We have
\begin{align*}
    \sum_{t=1}^T\tilde w_t - \sum_{t=T+1}^{2T}\tilde w_t&=\sum_{t=1}^Tw_t-\sum_{t=T+1}^\tau w_t - \sum_{t=\tau+1}^{2T}(\overline w+\beta)
    \\& = \sum_{t=1}^T(w_t - \overline w-\beta) - \sum_{t=T+1}^\tau(w_t - \overline w - \beta)
    \\& = \sum_{\substack{t \leq T \\ t \in H^+}}(w_t - \overline w - \beta) - \sum_{\substack{t>T \\ t \in H^+}}(w_t - \overline w - \beta)
    \\&+\sum_{\substack{t \leq T \\ t \notin H}}(w_t - \overline w - \beta) - \sum_{\substack{\tau>t>T \\ t \notin H}}(w_t - \overline w - \beta)
    \\& \geq \sum_{\substack{t \leq T \\ t \in H^+}}(w_t - \overline w - \beta) - \sum_{\substack{t>T \\ t \in H^+}}(w_t - \overline w - \beta)-2Th(T).
\end{align*}
The second equality follows because we can distribute the $2T-\tau$ subtractions of $\overline w+\beta$ to the first two summations. In the third equality we use again the fact that every $t \leq \tau$ and $t \in H$ belongs to $H^+$. The final inequality follows from the definition of $H$.

It remains to bound $\sum_{\substack{t \leq T \\ t \in H^+}}(w_t - \overline w - \beta) - \sum_{\substack{t>T \\ t \in H^+}}(w_t - \overline w - \beta)$. Write
\begin{align*}
    \sum_{\substack{t \leq T \\ t \in H^+}}(w_t - \overline w - \beta) - \sum_{\substack{t>T \\ t \in H^+}}(w_t - \overline w - \beta)&=2 \sum_{\substack{t \leq T \\ t \in H^+}}(w_t - \overline w - \beta) - \sum_{t \in H^+}(w_t - \overline w - \beta)
    \\&\geq \left(\frac{2|H^+ \cap \{1, 2, \dotsc, T\}|}{|H^+|}-1 \right) \sum_{t \in H^+}(w_t-\overline w -\beta)
    \\&\geq \left(\frac{2(T-2Th(T))}{2T-A+2Th(T)}-1 \right) \sum_{t \in H^+}(w_t - \overline w - \beta)
    \\&= \frac{A-6Th(T)}{2T-A+2Th(T)} \sum_{t \in H^+}(w_t - \overline w - \beta)
    =\Omega(T).
\end{align*}
Here the first inequality follows because $w_t-\overline w-\beta$ is decreasing and non-negative for $t \in H^+$. The second inequality uses $2T-A+2Th(T) \geq |H^+| \geq T-2Th(T)$ by \Cref{eqn:small_tau} and the definition of $H^+$. To see the final equality, observe that $A= \Omega(T)$ and hence $\frac{A-6Th(T)}{2T-A+2Th(T)}=\Omega(1)$, and by \Cref{eqn:H_plus_bound}, $\sum_{t \in H^+}(w_t - \overline w - \beta)=\Omega(T)$. This concludes the proof of the claim.

\subsection{\texorpdfstring{Proof of \Cref{lemma:nl_bound}}{Proof of Lemma~\ref*{lemma:nl_bound}}}\label{subsec:nl_bound_pf}
\begin{lemma*}
    We have
    \[ \sum_{t=1}^{2T} \Pr[a_t=L] \leq \frac{2\eta}{\beta}T+o(T). \]
\end{lemma*}

Write
\begin{align*}
    L_\text{act} = \ee{}{\sum_{t=1}^{2T} \ind{a_t=L}G_L^{(t)}}, \quad & M_\text{act} = \ee{}{\sum_{t=1}^{2T} \ind{a_t=M}G_M^{(t)}}, \\
    C_{L \to M} = \ee{}{\sum_{t=1}^{2T} \ind{a_t=L}G_M^{(t)}}, \quad & C_{R \to M} = \ee{}{\sum_{t=1}^{2T} \ind{a_t=R}G_M^{(t)}},
\end{align*}
and
\[ N_L = \sum_{t=1}^{2T} \Pr[a_t=L]. \]

By \Cref{lemma:er}, the full-horizon fixed-action payoffs are
\[ \sum_{t=1}^{2T} \ee{}{G_L^{(t)}}=-\frac A4, \quad \sum_{t=1}^{2T} \ee{}{G_M^{(t)}}=0, \quad \sum_{t=1}^{2T}G_R^{(t)}=0. \]
Thus the best fixed action is $M$ (or $R$), with payoff zero. By the assumption for contradiction, the column player has non-negative best-in-hindsight regret on the actual $2T$-round play up to $2 \eta T+O(TR(T))$
\begin{equation}\label{eqn:lact_mact_ner}
    L_\text{act}+M_\text{act} \leq 2\eta T+O(TR(T)).
\end{equation}
Now, swap regret for the full horizon swap $L \to R$ gives
\begin{equation}\label{eqn:lact_bound}
    0-L_\text{act} \leq O(TR(T)), \quad \text{or equivalently} \quad L_\text{act} \geq -O(TR(T)).
\end{equation}
Similarly, swap regret for the full-horizon swap $R \to M$ gives
\begin{equation}\label{eqn:crm_bound}
    C_{R \to M}-0 \leq O(TR(T)).
\end{equation}

Observe that, since $\ind{a_t=L} + \ind{a_t=M} + \ind{a_t=R}=1$, the fixed payoff from action $M$ decomposes according to the action actually played by the column player:
\begin{equation}\label{m_total_equality}
    C_{L \to M}+M_\text{act}+C_{R \to M} = \ee{}{\sum_{t=1}^{2T}G_M^{(t)}}=0.
\end{equation}
Combining \Cref{eqn:lact_mact_ner,eqn:crm_bound,m_total_equality}, we obtain
\begin{align}
    C_{L \to M}&=0-M_\text{act}-C_{R \to M} \nonumber
    \\& \geq 0-(2\eta T+O(TR(T))-L_\text{act})-O(TR(T)) \nonumber
    \\&=L_\text{act}-2\eta T-O(TR(T)). \label{eqn:clm_bound}
\end{align}
Here the first equality follows from \Cref{m_total_equality}. The first inequality follows by the upper bounds in \Cref{eqn:lact_mact_ner,eqn:crm_bound}.

We now note the special algebraic relation between $G_L^{(t)}$ and $G_M^{(t)}$. Recall that $U_L=1-\overline w - \beta, U_M = \overline w-1$ and $D_L=-\overline w - \beta, D_M=\overline w$, so $U_M=-U_L-\beta$ and $D_M=-D_L - \beta$. Therefore, for every round $t$
\begin{equation}\label{eqn:gmgl_relationship}
    G_M^{(t)}=-G_L^{(t)} - \beta.
\end{equation}
In particular, in the rounds in which the column player actually plays $L$, this gives
\begin{align}
    C_{L \to M} &= \ee{}{\sum_{t=1}^{2T} \ind{a_t=L}G_M^{(t)}} \nonumber
    \\&=-\ee{}{\sum_{t=1}^{2T} \ind{a_t=L}G_L^{(t)}} - \beta \sum_{t=1}^{2T} \Pr[a_t=L] \nonumber
    \\&=-L_\text{act} - \beta N_L. \label{eqn:clm_bound_nl}
\end{align}
Here, for the second equality, we use \Cref{eqn:gmgl_relationship} and $\ee{}{\ind A} = \Pr[A]$ for an event $A$.

Combining \Cref{eqn:clm_bound,eqn:clm_bound_nl}, we obtain
\[ C_{L \to M}=-L_\text{act} - \beta N_L \geq L_\text{act}-2\eta T-O(TR(T)), \]
which rearranges into
\[ \beta N_L \leq -2L_\text{act}+2\eta T+O(TR(T)). \]
By \Cref{eqn:lact_bound} $L_\text{act} \geq -O(TR(T))$. Multiplying both sides by $-2$ gives $-2L_\text{act} \leq 2 \cdot O(TR(T))=O(TR(T))$. Hence $\beta N_L \leq 2\eta T+O(TR(T))$. Now, recall that we have shown $C' \leq \beta \leq \frac 14$ for all sufficiently large $T$ (in \Cref{eqn:beta_theta1}), so we can divide both sides by $\beta = \Theta(1)$ to obtain
\[ N_L \leq \frac{2\eta}{\beta}T+ O(TR(T)). \]
The lemma follows by noting that $R(T)=o(1)$ and $N_L \geq 0$.

\end{document}